\documentclass{scrartcl}

\usepackage[english]{babel}
\usepackage[latin1]{inputenc}
\usepackage{amsmath,amsthm,amsfonts,amssymb,times}

\usepackage[final]{epsfig}
\usepackage{color}

\usepackage{mathrsfs}

\newtheorem{theorem}{Theorem}[section]

\newtheorem{lemma}[theorem]{Lemma}

\theoremstyle{definition}
\newtheorem{definition}[theorem]{Definition}

\numberwithin{equation}{section}

\usepackage{graphicx}

\makeatletter

\let\Re\undefined

\DeclareMathOperator{\Re}{Re \,}

\DeclareMathOperator{\tr}{tr}

\newcommand{\Z}{ {\mathbb Z} }
\newcommand{\N} {{\mathbb N}}
\newcommand{\C}{\mathbb{C}}
 
\newcommand{\G}{\mathcal{G}}
\newcommand{\R}{\mathbb{R}}

\def\t#1{\widetilde #1}

\DeclareMathOperator{\indfct}{\rm 1}

\def\be{\begin{equation}}
\def\ee{\end{equation}} 
\def\ba{\begin{eqnarray} }
\def\ea{\end{eqnarray} } 

\def\={\  =  \  } 

 \DeclareMathOperator{\Pf}{Pf}

\DeclareMathOperator{\sgn}{sign}

\DeclareOldFontCommand{\rm}{\normalfont\rmfamily}{\mathrm}
\DeclareOldFontCommand{\sf}{\normalfont\sffamily}{\mathsf}
\DeclareOldFontCommand{\tt}{\normalfont\ttfamily}{\mathtt}
\DeclareOldFontCommand{\bf}{\normalfont\bfseries}{\mathbf}
\DeclareOldFontCommand{\it}{\normalfont\itshape}{\mathit}
\DeclareOldFontCommand{\sl}{\normalfont\slshape}{\@nomath\sl}
\DeclareOldFontCommand{\sc}{\normalfont\scshape}{\@nomath\sc}

\makeatother

\begin{document}

\title{\Large \bf 
Kac-Ward formula and its extension to  order-disorder correlators through a graph zeta function
}

\author{\normalsize Michael Aizenman\footnote{Departments of Physics and Mathematics,  Princeton University,  Princeton NJ 08544, USA. }
\and  \normalsize
Simone Warzel\footnote{Zentrum Mathematik, TU M\"unchen, 
Boltzmannstr. 3, 85747 Garching, Germany. }
}

\date{ \flushright \normalsize  {\it  To Uzy Smilansky in celebration of his 
77th birthday}
\\[4ex]
 \small  \today  
 \hspace{5cm} \mbox{} \\[-5ex] }


\maketitle

%




\begin{abstract} 
A streamlined derivation of the Kac-Ward formula for the planar Ising model's partition function  is  presented and applied in relating the kernel of the Kac-Ward matrices' inverse  with the correlation functions of the Ising model's order-disorder correlation functions.  
A shortcut for both  is facilitated by the Bowen-Lanford graph zeta function relation.  The Kac-Ward relation is also extended here to produce  a family of non planar interactions on $\Z^2$ for which the partition function and the order-disorder correlators are solvable at special values of the coupling parameters/temperature.
 

\end{abstract}


\section{Preamble} 

The focus of this note is on the application of a graph zeta function for a  simplified Kac-Ward style derivation of Onsager's formula for the free energy of the planar Ising model.   This approach 
also adds  insight into the origin of some of the model's astounding properties and emergent fermionic structures. 
More specifically: we i)   present a simple derivation of  the Kac-Ward formula for the free energy~\cite{KW52}, for which the relation to a graph zeta function is put to use, ii)  point out a partial extension of the solvability via the Kac-Ward determinant to a class of only ``quasi-planar'' models, at temperature dependent coupling constants, 
iii) reaching beyond the partition function,  spell out  a direct relation between the Kac-Ward matrix' resolvent  
and the correlation function of the model's order-disorder  operators.

The subject has a rich history some of whose highlights are mentioned below,  within its relevant context.  
In particular, this paper bears relation with the recent works of 
Cimasoni~\cite{Cim10} and Lis~\cite{Lis15}.   The latter includes a short combinatorial proof of the Kac-Ward determinantal formula for the Ising model's partition function.  The proof given here is closely related, but we try to further elucidate the link with the graph zeta function.  We also  identify an algebraic statement which allows  to reduce the combinatorics to what is essentially contained  in Kac-Ward's first step.  

The Kac-Ward matrix was also noted to play a role in some of the model's fermionic variables; e.g., in~\cite{Lis14} (and references therein), the matrix' inverse  kernel was identified, through a  natural path expansion,  as a fermionic generating function.   To this discussion we add a more explicit recognition of the role of this kernel in the correlation function of the model's order-disorder operators.    These are related to, but not the same as  the more recent Smirnov's parafermionic  amplitude, whose links  with the Kac-Ward matrix are discussed in~\cite{Cim12,KLM13}.

 We do not specifically  address here the Ising model's critical point,    on which much has been said  in the literature  including from the Chelkak-Smirnov perspective of the model's s-holomorphicity~\cite{Smi10,CS12}.

 \section{Onsager's solution}   
  
  The Ising model is a system of $\pm1$ valued \emph{spin variables} $(\sigma_x)$ attached to the vertices $ \mathcal{V} $ of a  graph     $\G = (\mathcal{E}, \mathcal{V}) $ 
  with the energy function 
\be \label{H_Ising}
H(\sigma)  \  =\   - \sum_{\{x,y\} \in \mathcal E} J_{x,y} \sigma_x \sigma_y  - h \sum_{x\in \mathcal{V}}  \sigma_x
\ee 
given in terms of edge couplings $ (J_{x,y}) $.
The probability distribution representing thermal equilibrium is the Gibbs measure, referred to as the \emph{Gibbs state}, 
\be 
{\rm Prob}( \sigma) \ =\  e^{-\beta H(\sigma)}/ \, Z_\G(\beta,h) 
\quad \mbox{ with
 $Z_\G(\beta,h)  \ = \ \sum_\sigma e^{-\beta H(\sigma)}$, } 
\ee 
for which the expectation value will be denoted by $ \langle \cdot \rangle_{\beta, h}^{\G} $ (some of whose indices will occasionally be omitted).
The normalizing factor is  the {\it partition function} $Z_\G(\beta,h) $.    Through it, one computes the \emph{thermodynamic pressure}:  \be 
{\psi(\beta,h) \ = \  \frac{1}{|\mathcal{V}|}   \log Z_\G(\beta,h)}  \, .
\ee 
Of particular interest is the infinite volume limit $ |\mathcal{V}| \to \infty $ for the pressure, for which singularities  may develop corresponding to phase transitions.
 Of further  interest are the infinite volume limits of the Gibbs equilibrium states.  These can be viewed as tangents of the pressure functional~\cite{Rue}, and unlike $\psi(\beta,h)$ their limits may depend on the boundary conditions.

For most of the analysis which follows, the model's coupling constants  need not be  constant, or ferromagnetic.   The first of these restrictions will be invoked only in extracting  Onsager's formula for the free energy from  the  Kac-Ward determinantal expression  for the case $\G = \Z^2$ and $J_{x,y} = J$ for all $\{x,y\} \in \mathcal E(\Z^2)$.   The field  of statistical mechanics was  transformed by Onsager's calculation of the dependence of the Ising model's   free energy  on temperature, at zero magnetic field.

\begin{theorem}[Onsager \cite{Ons_1}]
For the Ising model with  constant nearest-neighbor couplings $J_{x,y} = J >0$  on  $\mathcal{G} = \Z^2$  the infinite-volume pressure 
is given by 
\begin{align}\label{eq:Onsager}
 \psi(\beta,0) = & \ \frac{1}{2} \int \ln \left\{ \left[Y(\beta)+ Y(\beta)^{-1} - 2 \right]+ E(k_1,k_2) \right\} \frac{dk_1 dk_2}{(2 \pi)^2}   
 + \, \frac{1}{2}   \ln  \cosh( 2\beta J)  
\end{align}
with 
\be
   Y(\beta) \ :=\   \sinh(2\beta J) \, , \qquad E(k_1,k_2) \ :=\  2 \sin^2(\frac{k_1} 2 ) + 2 \sin^2(\frac{k_2}2) \, . \notag
\ee 
\end{theorem}


The significance and impact of Onsager's solution cannot be overstated.   It  demonstrated for the first time the possibility of modeling the 
divergence of the specific heat at a phase transition; in this case   as $C \log |\beta-\beta_c|$, with $\beta_c$ characterized by $Y(\beta_c) =1$.

 Solutions of other 2D models followed, though none  in $D>2$ dimension.  Despite the lack of exact solutions many of the essential features of the model's critical behavior have been understood (including at the level of rigorous results, which this may not be the place to list) and more may be on its way~\cite{EPPRS14}.

Onsager's solution was algebraic in nature.  Its different elaborations and simplification demonstrate
the utility  in this context of a number of anti-commuting structures, which show in various forms.  Some of these are listed here in Section~\ref{sec:KadCev}, within the context of the relations discussed here.

\section{Graph zeta functions and Kac-Ward matrices} 

\subsection {Paths on graphs with transition-based weights}

For a graph,   described in terms of its vertex and edge set as $\G = \left(\mathcal{V} , \mathcal E \right)$, the set of its   \emph{oriented edges} will be denoted here by  $ \widehat {\mathcal E}  \equiv \widehat {\mathcal E} (\G) = \mathcal E \times \{+1, -1\}$.      
 These will be used  to describe paths and loops on $\G$. 
The presentation of paths in this form, 
 rather than through  sequences of visited sites,  simplifies the formulation of weights which depend multiplicatively on both the edges visited and turns taken.   

\begin{definition} (Paths and loops --  basic terminology and notation) 
  \begin{enumerate}
\item An oriented edge $e \equiv (x_1,x_2) $ is said to lead into  $e' \equiv (y_1,y_2)$ if $y_1=x_2$.   The relation is denoted  $e \rhd e'$.  
The reversal of $e$ is denoted $\overline e$. 
\item  A path on $\G$  is  a finite sequence $ \widehat \gamma = (e_0, e_1 , \dots , e_n ) $ of  oriented edges sequentially leading into each other.  The number of such steps, $ n $, is denoted by $ | \widehat \gamma | $.   
 A path is said to be 
backtracking  if $e_{j+1} =  \overline e_j$ for some $ j $. 
\item  
A loop is an equivalence class, modulo cyclic reparametrization, of non-backtracking paths $ (e_0, e_1 , \dots , e_n ) $ with $e_{n} = e_0$. 
A loop is called \emph{primitive} if it  cannot be presented as some $k$-fold repetition  of a shorter one. 
 \item A $ \widehat {\mathcal E} \times  \widehat {\mathcal E} $  matrix $\mathcal M$  is  a flow matrix if $ \mathcal M_{e,e'} \ = \ 0$  unless  $e \rhd e'$. It is called \emph{non-backtracking} if $\mathcal M_{\overline {e} ,e} \ = \ 0$   for all $ e\in \widehat {\mathcal E}$. For a $ \widehat {\mathcal E} \times  \widehat {\mathcal E} $  matrix $\mathcal M$ and a 
  path or loop $ \widehat \gamma = (e_0, e_1 , \dots , e_n ) $,  we denote: 
\be \mathcal  \displaystyle  \chi_{_{\mathcal M}} ( \widehat \gamma) := 
  \prod_{j = 1}^n  \mathcal M_{e_{j-1},e_{j}} \, . 
  \ee
(which for loops is invariant under cyclic reparametrizations). 
\end{enumerate}
 
 \end{definition}
 
The following general results allows a natural introduction of  the Bowen-Lanford zeta function $\zeta_M$~\cite{BL70}.

\begin{theorem} \label{thm:Ihara} Let  $\mathcal{M}$ be a non-backtracking  $\widehat{\mathcal E}\times \widehat{\mathcal E} $ flow  matrix over a finite graph. 
Then  for all  $u \in \C$ such that 
$  \label{eq:norm_M}
|u|^{-1}  >  \| \mathcal{M} \|_{\infty} \ := \  \max_{e\in \widehat{\mathcal{E}}}   \sum_{e'\in\widehat{\mathcal{E}}} |\mathcal{M}_{e,e'}| \,      
$:
\be  \label{eq:zetaM}
  \det (1 - u\mathcal{M}) \ =\  \prod_{p} \left[ 1 - u^{|p|}\chi_{\mathcal{M}}(p) \right] \quad  \left[=: \, \zeta_M(u)^{-1}\right]
\ee 
where the product ranges over the 
primitive oriented loops, each equivalence class contributing a single factor.   
\end{theorem}

\begin{proof}  
For $u $ small enough 
$\ln \det (1 - u\mathcal{M}) $ is analytic in $u$ and satisfies: 
\be \label{ln_det}
\ln \det (1 - u\mathcal{M}) \  = \tr \ln (1 - u\mathcal{M}) = - \sum_{n=1}^\infty \frac{u^n}{n} \tr  \mathcal{M}^n \,    .   
\ee 
The sum is convergent since $| \tr \mathcal{M}^n | \leq \| \mathcal{M} \|_{\infty}^n$.  Furthermore, the assumed bound 
$|u| < \| \mathcal{M} \|_{\infty} ^{-1}$ 
 also guarantees the absolute convergence of   the sums obtained  by splitting:
\begin{eqnarray} 
\tr  \mathcal{M}^n &=& \sum_{e_1,\dots,e_n \,   \in \widehat{\mathcal E}}  \, \mathcal{M}_{e_n,e_1} \times \prod_{j=1}^{n-1}  \mathcal{M}_{e_{j},e_{j+1}}  \\  
&=& \, \sum_{e\in \widehat{\mathcal E}} \sum_{\substack{ \widehat \gamma \ni e \\  |\widehat\gamma| = n}}   \,  \chi_{\mathcal{M}}( \widehat \gamma) 
\ = \ 
 \sum_{k=1}^\infty \,  \, \sum_{p}  \delta_ {k |p|, n}  \, \, |p|     \; 
 \chi_{\mathcal{M}}(p)^k  \, ,  \notag
\end{eqnarray} 
where the last equation is obtained by representing each  loop $ \widehat \gamma$ as the suitable power ($k$) of a primitive oriented loop and $p$ ranges over equivalence classes of such,  modulo cyclic permutation.  
Thus  
\be
\sum_{n=1}^\infty \frac{u^n}{n} \tr  \mathcal{M}^n = \sum_{k=1}^\infty \sum_{p}  \frac {1}{ k} \, \left [ u^{|p|}\chi_{\mathcal{M}}(p) \right] ^k = - \sum_{p}  \ln \left( 1 - u^{|p|}\chi_{\mathcal{M}}(p) \right) \, .  
\ee
Combined with \eqref{ln_det}  this yields \eqref{eq:zetaM}.  
\end{proof}

The meromorphic function $\zeta_M$ defined in \eqref{eq:zetaM}  is related to Ihara's graph zeta function, whose usual definition is in terms of site-indexed paths and $\mathcal{V}\times \mathcal{V}$ matrices.  In that case the customary non-backtracking restriction on $p$  needs to be added explicitly.   Also the statement and proof of the corresponding result are a bit more involved.  Further background on this topic, other extensions, and graph theoretic applications, can be found in, e.g., \cite{ST96,Smi09}.

In view of the simplicity of the argument, it may be worth stressing that the  limitations placed on the loops in \eqref{eq:zetaM}  leave the possibility of self intersection, multiple crossing of  edges, and arbitrary repetitions of  sub-loops.  The product in \eqref{eq:zetaM}  is over an infinite collection of factors, whose series expansions  yield terms with arbitrary powers of $u$.   However on the left side is a polynomial in  $u$.  Thus contained in  this zeta function relation is an infinite collection of combinatorial cancellations reflecting the fact that the loop ensemble has a fermionic nature. \\

More can be said in case of flow matrices  which obey the following symmetries.
\begin{definition}
A  $\widehat{\mathcal E}\times \widehat{\mathcal E} $ flow  matrix $  \mathcal{M} $ is said to be
\begin{enumerate}
\item \emph{loopwise time-reversal invariant} if 
\be \label{eq:looptrinv}
 \chi_{_{\mathcal M}} ( \widehat \gamma)  \ = \  \chi_{_{\mathcal M}} ( \widehat \gamma_R )  
\ee
for any oriented loop $ \widehat \gamma = (e, e_1 , \dots , e_{n-1}, e) $ and its time-reversed oriented  loop  $ \widehat \gamma_R :=   (\overline{e} , \overline{e}_{n-1} , \dots , \overline{e}_{1} , \overline{e} ) $.
\item  \emph{twist anti-symmetric} if 
\be\label{eq:looptwist}
 \chi_{_{\mathcal M}} ( \widehat \gamma)  =  - \chi_{_{\mathcal M}} ( \widehat \gamma_T )
\ee
for any oriented path $ \widehat \gamma =   (e , e_1 , \dots , e_{n-1} , \overline{e} ) $, 
 which starts on some edge $e$ and ends on the reversed edge $\overline e$,  and its uniquely associated twist path  
$ \widehat \gamma_T := (e , \overline{e}_{n-1} , \dots , \overline{e}_{1} ,   \overline{e} ) $. 
\end{enumerate}
\end{definition}

In order to state the special properties of zeta functions of such matrices, it is convenient to associate with each oriented edge a weight, $ W_e = W_{\overline e} \neq 0  $, which is independent of the orientation,  and 
an associated unoriented weight matrix
\be
\mathcal{W} = {\rm diag}\left(W_e  \right) \, .
\ee
If $ \mathcal{K}  $ is a flow matrix which is loopwise time-reversal invariant or twist anti-symmetric, so will be the product
$ \mathcal{M} = \mathcal{K} \mathcal{W} $.

\begin{lemma} \label{lem:sq}
Let $\mathcal{K}$ be a non-backtracking, loopwise  time-reversal invariant and twist anti-symmetric $\widehat{\mathcal E}\times \widehat{\mathcal E} $ flow  matrix over a finite graph. 
Then for any unoriented weight matrix $ \mathcal{W} $, the determinant
$\det(1 - \mathcal{ K W}) $ is the square of a multilinear polynomial in  the weights $(W_e )$.  
\end{lemma} 

\begin{proof}  
Since $\det (1 - \mathcal{M}) $ with  $ \mathcal{M} = \mathcal {K W}  $ is a quadratic polynomial in each  $ W_e $, it suffices by analyticity to establish the result in the case that all $(W_e )$ are small, i.e., if
\be \label{eq:asssmallW}
 \|  \mathcal{M} \|_\infty \ \leq \ \max {\rm degree}(\G) \; \times \max_e | W_e |  \ < 1 \, .
\ee 
The dependence of the matrix $ \mathcal{M}   $ on the parameter $ W_e $ for a fixed  $e$ is linear with derivative 
\be 
D^{(e)} \ : = \  \frac{\partial \mathcal M}{\partial \mathcal W_e} \ =\  \mathcal K \, P^{(e)}
\ee 
with the rank-two matrix  
$ P^{(e)} := | \delta_{e} \rangle \langle  \delta_{e} | + |\delta_{\overline{e}} \rangle \langle \delta_{\overline{e}} | $, 
abbreviating the orthogonal projection onto the subspace spanned by the (delta) functions supported only on $ e$ and $\overline{e} $.
The assumption~\eqref{eq:asssmallW} allows us to apply standard formulas for the derivatives of the logarithm of a determinant: 
\begin{multline}\label{eq:diffdet}
\frac{\partial^2}{\partial \mathcal W_e^2} \sqrt{\det (1 - \mathcal{M}) } =   
 \frac{1}{2}   \sqrt{\det (1 - \mathcal{M}) } \\
\times  \left \{ \frac{1}{2}  \left [ \tr\left( \frac{1}{1 - \mathcal{M}} \mathcal{D}^{(e)} \right) \right ] ^2 - \tr\left(  \frac{1}{1 - \mathcal{M}} \mathcal{D}^{(e)} \frac{1}{1 - \mathcal{M}} \mathcal{D}^{(e)} \right) \right \} \, . 
 \end{multline}
Abbreviating    
\be\label{eq:repA}
A := P^{(e)} \frac{1}{1 - \mathcal{M}} \mathcal{D}^{(e)} P^{(e)} \ =\ 
 P^{(e)} \sum_{n=1}^\infty  \mathcal{M}^n \mathcal{K} P^{(e)}
\ee
and viewing it as a $ 2 \times 2 $ matrix,  the bracket 
$ \{ \dots \} $  in~\eqref{eq:diffdet} equals
\be 
\frac{1}{2}  (\tr A)^2 - \tr A^2 = - \frac{1}{2} \, {\rm Discriminant}(A) \, . 
\ee 
From the power series representation in~\eqref{eq:repA}  and properties~\eqref{eq:looptrinv} and~\eqref{eq:looptwist}, we conclude that $A$ is a multiple of the $2\times 2$ identity matrix.  Hence its   discriminant is zero. This implies  that $ \sqrt{\det (1 - \mathcal{M}) }$ is linear in each $ W_e$ as  claimed. 
\end{proof}

\subsection{Kac-Ward matrices}

A key role in  Kac' and Ward's~\cite{KW52} combinatorial solution of the Ising model is played by the $\widehat{\mathcal E}\times \widehat{\mathcal E} $ non-backtracking matrix $ \mathcal{K} $, which for planar graphs  of straight edges is defined as 
\be  \label{eq:defK}
 \mathcal{K}_{e',e} \ :=  \  \indfct[ e' \rhd e; e'\neq \overline{e}]  \,  \exp\left(\tfrac{i}{2} \angle(e,e')\right)   \, , 
\ee 
with  $ \angle(e,e') \in (-\pi, \pi)$ the 
difference in the argument of the tangent of $e'$ relative to the tangent of $e$ at the vertex at which the two edges meet.  
To increase the method's reach, it is convenient to consider also an extension of   the 
corresponding matrix beyond strict planarity, to  faithful projections of graphs to $\R^2$. which we define as follows. (An example is depicted in Fig.~\ref{fig_3D}.)

\begin{definition}  A faithful projection of a graph $\G$ (which need not be planar)  to $\R^2$ is a graph drawn in the plane such that 
\begin{enumerate} 
\item the projection is graph isomorphic to $\G$, in the sense that  pairs of projected edges meet at a common end point only if so do the corresponding edges in~$\G$,   
\item the projected edges are described   by piecewise differentiable simple curves  which except for their end points avoid the graph vertices. 
\end{enumerate} 
\end{definition}

\begin{figure}[h]
 \begin{center} 
\includegraphics[width=.4 \textwidth]{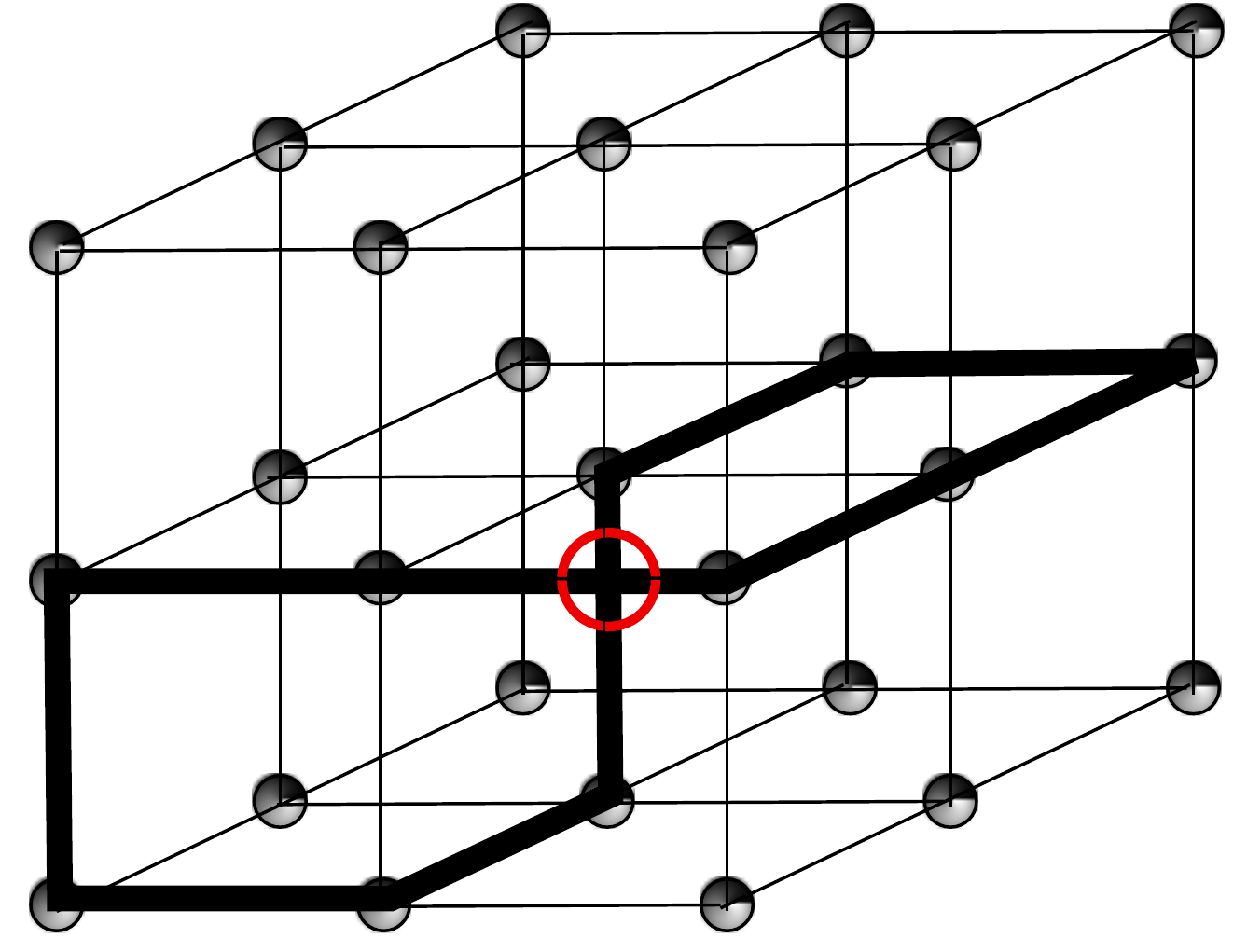}
\end{center}
\caption{A faithful projection on $\R^2$ of a non-planar graph.  Highlighted is a loop whose projection exhibits a non-vertex edge crossing, i.e., one which is not realized in the original graph.  
Such crossings are excluded in the planar projections which are discussed in Theorem~\ref{thm:KW}.  Their effects is presented in Theorem~\ref{thm:beyondP}.}  \label{fig_3D}
\end{figure}

We extend the definition of the matrix~\eqref{eq:defK} to such faithful projections    by 
letting   $ \angle(e',e) \in (-\pi, \pi]$ be  the sum of the discrete changes in  the tangent's argument, from the origin of $e$ to that of $e'$, plus the integral of the increase in the tangent's argument increments along the edge $e$.
Referring to planar graphs  we shall by default mean graphs faithfully projected to
$\R^2$ 
with no-crossing edges. \\

A key property of the Kac matrix $\mathcal {K}$ defined in~\eqref{eq:defK} is that  the product of the edge weights associated with $\mathcal {K}$ along a primitive oriented loop $ p $ imbedded into the plane $ \R^2 $ yields the parity of the winding $ {\rm win}(p) $ of the loop's tangent.   This in turn is related by Whitney's theorem~\cite{Whit}  to the 
parity of the loop's self crossing number $ n(p) $: 
\be \label{Whitney}
\chi_{_\mathcal{K }}(p)  \ =\  (-1)^{{\rm win}(p)}     \ = \  - (-1)^{n(p)} \, .
\ee 
(For planar loops the self crossing number's parity $ (-1)^{n(p) }$ is  well defined even for loops  which repeat a bond.  It is  the asymptotically common value of the corresponding quantity for all approximating sequences  in which $p$ is  approximated, in the natural distance, by smooth loops in $\R^2$ with only transversal self crossings.)

Consequently, $\mathcal {K}$ is not only non-backtracking, but also loopwise time-reversal invariant. Moreover, since any oriented path of the form
 $ \widehat \gamma =   (e , e_1 , \dots , e_{n-1} , \overline{e} ) $ can be closed to a loop in the plane (by connecting the endpoint of $  \overline{e} $ to the starting point of $ e $), 
the corresponding twisted path $ \widehat \gamma_T $ has an additional self-crossing and hence $\mathcal {K}$ is also twist anti-symmetric. 
Lemma~\ref{lem:sq} then allows the following shortcut in the zeta-function representation of the principal value of the square root of the 
determinant $  \det (1 - \mathcal{K W}) $ in terms of 
equivalence classes of \emph{unoriented loops}. For any primitive loop $ \gamma = (e_0, e_1, \dots , e_n = e_0 ) $ with $ e_j \in \mathcal{E} $, we associate
the weight
\be
w(\gamma) = \prod_{j=1}^n W_{e_j} \, . 
\ee

\begin{lemma} \label{lem:max_KW}
For any faithful projection of a finite graph $ \G $ to $ \R^2 $ and supposing~\eqref{eq:asssmallW}:
 \be \label{eq:max_KW}
\sqrt{ \det (1 - \mathcal{K W}) } \ =   \  1+ 
 \sum_{n=1}^\infty 
 \sum_{\substack{\gamma_1,\dots, \gamma_n\\ \|\{\gamma_1,..., \gamma_n\}\|_\infty   = 1}}
 \prod_{j=1}^n   
(-1)^{n(\gamma_j)} w(\gamma_j) \, . 
\ee 
The sum extends over equivalence classes of unoriented loops $ \gamma_1, \dots , \gamma_n $ on $ \mathcal{E} $ in which we restrict to the case that the maximal number of times an edge is covered by the given collection of loops is one, 
\be
\|\{\gamma_1,..., \gamma_n\}\|_\infty := \max_{e \in \mathcal{E}  } \sum_{j=1}^n \sharp [ e \in \gamma_j ] \,   ,
\ee
 with $ \sharp [ e \in \gamma_j ]  $ the number of times the edge $ e $ is traversed by the loop $ \gamma_j $.  
\end{lemma} 

\begin{proof}  
The zeta function relation~\eqref{eq:zetaM} applied to the 
the Kac-Ward matrix $ \mathcal{K}\mathcal{W} $ gives for $  |u| $ small enough
\begin{align}   \label{212}
\det (1 - u\, \mathcal{K W})  & = \    \prod_{p} \left[ 1 - u^{|p|} \, \chi_{\mathcal KW}(p)\ \right]  \\   
&  = \ 
 \prod_{p} \left[ 1 + u^{|p|} (-1)^{n(p) } w(p) \right]  
\ = \  
  \prod_{\gamma} 
  \left[ 1 + u^{|\gamma|} (-1)^{n(\gamma) }w(\gamma) \right]^2 \notag
\end{align}  
where we used: i) the relation
$
\chi_{_\mathcal{KW }}(p) \ = \ \chi_{_\mathcal{K }}(p)  \, w(p) $, ii)  the Whitney relation~\eqref{Whitney}, and iii)  
the fact   that both $(-1)^{n(p) } $ and  $w(p)$ are invariant under the loop's orientation reversal. 
Therefore, the contribution of every primitive oriented loop $ p $ equals the square of the corresponding primitive unoriented loop $ \gamma $.
Consequently,
\be \label{SQR}
\sqrt{ \det (1 - u \, \mathcal{K W}) } \  = \  
  \prod_{\gamma} 
  \left[ 1 + u^{|\gamma|} \, (-1)^{n(\gamma) }  \chi_{_\mathcal W}(\gamma) \right] 
\ee 
where $\gamma$ ranges over equivalence classes of \emph{unoriented} non-backtracking  loops.     
 The assumption~\eqref{eq:asssmallW} also allows to expand the product on the right side into 
\begin{align}\label{eq:expandsum}
\prod_{\gamma } \left[ 1+ u^{|\gamma|} (-1)^{n(\gamma)} \, w(\gamma) \right] & = 1 + \sum_{n=1}^\infty \sum_{\gamma_1,\dots, \gamma_n}  \prod_{j=1}^n   u^{|\gamma_j|} 
 (-1)^{n(\gamma_j)} \, w(\gamma_j) \,. 
\end{align}
Grouping the terms in the resulting sum according to  
$
\|\{\gamma_1,..., \gamma_n\}\|_\infty 
$ we  reach a key step in the argument: By  Lemma~\ref{lem:sq}, 
the left side of \eqref{SQR} is a multilinear function of the edge weights $(W_e)$.   By the uniqueness of power series expansion in the regime of its absolute convergence, one may deduce that the sum over terms on the right side of \eqref{SQR}    with 
 $ \|\{\gamma_1,..., \gamma_n\}\|_\infty   >  1$ vanishes.  
One is therefore left with only the terms with no doubly covered edges. 
Since there are only finitely many terms, the relation extends by analyticity to all $ u $, with $ u = 1 $ corresponding to the case claimed in~\eqref{eq:max_KW}. 
 \end{proof}

It should be noted that the last result  represents a major reduction of combinatorial complexity: the  sum in \eqref{eq:max_KW} is not only finite but it ranges over only edge disjoint collections of loops.

\section{The Kac-Ward solution}

 We now turn to the method proposed by Kac and Ward~\cite{KW52}  for a combinatorial solution of Ising model on planar graphs $ \G = (\mathcal{E}, \mathcal{V} ) $. 
 In discussing the model's partition function it is convenient to split off a trivial factor and denote: 
\be 
\widetilde Z_\G(\beta,h)  \  := \   Z_\G(\beta,h)  \  \ {\big /}  \left[ 2^{|\mathcal{V}|}  \prod_{\{x,y\}\subset  \mathcal{E}} \cosh (\beta J_{x,y} )\right ]   \,  .
\ee
The transition from  the determinantal expression  \eqref{eq:KW} to  Onsager's explicit formula of the partition function \eqref{eq:Onsager} for $ \mathbb{Z}^2 $ is through a standard calculation.  For completeness, it is included here in the Appendix.   

\begin{theorem}[Kac-Ward-Feynman-Sherman]     \label{thm:KW}  For any finite planar  graph $\G$:
\be    \label{eq:KW}
\widetilde   Z_\G(\beta, h=0) 
     \, = \,       \sqrt{\det  (1 - \mathcal{KW})}    
\ee 
with  $\mathcal W$ the $ \widehat{ \mathcal{E} } \times \widehat{ \mathcal{E} } $ weight matrix  with diagonal entries given by 
\be  W_{(x,y)} = W_{(y,x)}  =   \tanh(\beta J_{x,y}) \, .
\ee 
The right hand side involves the principal value of the square root function. 
\end{theorem}

Theorem~\ref{thm:KW} has a somewhat long and often restated history.  The formula  was presented in the  proposal by Kac and Ward~\cite{KW52} for a combinatorial solution of the Ising model which does not require the algebraic apparatus of Onsager~\cite{Ons_1} and Kaufman~\cite{Kauf49}.   Subsequently, gaps of formulation were pointed out by Feynman,  who  proposed a systematic approach to an order by order derivation of this relation, in steps of increasing complexity~\cite{Feyn}.  The full proof was accomplished by Sherman~\cite{She60} and later also by Burgoyne~\cite{Bou63}, who organized the combinatorial analysis to all orders.  At the same time the determinantal formula was given other derivations by Hurst and Green~\cite{HG_60}, Kasteleyn~\cite{Kas63} and Fisher~\cite{Fisher_Ising}  based on Pfaffians and relations to the dimer model. 
Proofs were extended to arbitrary planar graphs on orientable surfaces in which the connection to spin structures and graph zeta functions was already made;
see also\ \cite{Cim10,Lis15} and references therein.

In the remainder of this section we give a short proof of Theorem~\ref{thm:KW} based on the results on graph zeta functions for Kac-Ward matrices. As usual, this derivation
starts from the following graphical high temperature representation.

\subsection{A graphical high temperature representation}

For any collection of edges $\Gamma \subset  \mathcal{E} $, we denote by $\partial \Gamma$ 
  the collection of vertices of the graph which are oddly covered by $\Gamma$.   Edge collections with $\partial \Gamma = \emptyset$  correspond to so called \emph{even subgraphs} of $\G$.  
In these terms one has: 

\begin{lemma} \label{lem:evenGraph} For any finite graph $ \G $ the Ising partition function at vanishing field admits  the following expansion into even subgraphs: 
\begin{eqnarray}  \label{Z_graph}
\widetilde Z_\G(\beta,h=0)  \ = \ 
  \sum_{\Gamma \subset \mathcal E\, : \, \partial \Gamma = \emptyset } w(\Gamma) \,    
\end{eqnarray} 
with   \, 
$w(\Gamma) := \prod_{\{x,y\} \in \Gamma }   W_{\{x,y\} }$ 
at the weights 
\be \label{eq:weights}
W_{\{x,y\} } \ = \   \tanh(\beta J_{x,y})  \,.  
 \ee
 The corresponding spin correlation functions of any even number of vertices $ A \subset \mathcal{V} $ are represented by
 \be\label{eq:spincorr}
  \langle \prod_{x \in A} \sigma_x \rangle_{\beta, h=0} \;  \widetilde Z_\G(\beta,h=0) \, 
 \ = \  \sum_ {\Gamma \, : \, \partial \Gamma = A } 
   w(\Gamma) \, . 
\ee
 \end{lemma}

This well known representation is obtained by writing
$$ 
e^{\beta J_{x,y} \sigma_x \sigma_y} 
=  \cosh( \beta J_{x,y} )  \left( 1 +  \sigma_x \sigma_y \tanh( \beta J_{x,y}) \right) ,
$$
expanding 
$$
\prod_{\{x,y\} \subset  \mathcal{E}}  \left( 1 +  \sigma_x \sigma_y \tanh \beta J_{x,y} \right) = \sum_{\Gamma } \prod_{\{x,y\} \in \Gamma}  \left [ \sigma_x \sigma_y \tanh( \beta J_{x,y}) \right   ] \, , 
$$ 
and then summing the resulting expression over the spin configurations $(\sigma)$. Equation~\eqref{Z_graph} 
results from the relation  $ \sum_\sigma \prod_{\{x,y\} \in \Gamma}  \sigma_x \sigma_y = 2^{|\mathcal{V}|} \, 1[\partial \Gamma = \emptyset] $ with $ 1[\dots] $ denoting the indicator function.
The second identity~\eqref{eq:spincorr} follows by similar arguments; see also~\cite{G06,DCbook}.

\subsection{A short proof of the Kac-Ward-Feynman-Sherman theorem} \label{sec:short_proof} 

The Ising model's partition function  is linked to the sum in \eqref{eq:max_KW} through the  following combinatorial identity
which, in the context of graphs of degree $4$, 
was among the insightful observations in  the Kac-Ward seminal paper.  
 
 \begin{lemma} [\cite{KW52}] \label{lem:KW_step1} 
For any  finite planar graph $ \G $  and any even subgraphs $\Gamma \subset  \mathcal{E} $ ($\partial \Gamma= \emptyset $)
\be   \label{eq:KW_step1}
  \sum_{ \Gamma = \sqcup \{\gamma_j\}   } 
  \prod_j (-1)^{n(\gamma_j)}
   = \  
 1 \, , 
\ee
where the sum is over all decompositions of $\Gamma = \sqcup \{\gamma_j\} $ into unoriented primitive loops.
\end{lemma} 
\begin{figure}[h]
 \begin{center} 
\includegraphics[width=.4\textwidth]{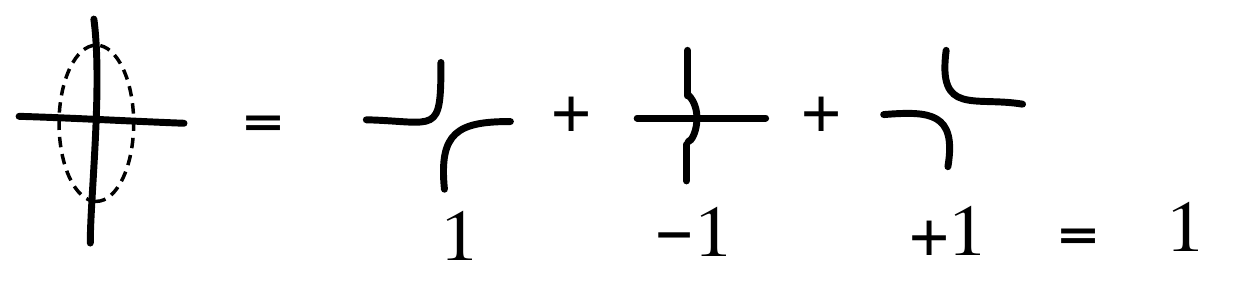}
\end{center}
\caption{The different  resolutions of a vertex of order $4$, which together determine the admissible loop decompositions of an even subgraph  $\Gamma$.  At each vertex the corresponding  $(\pm1)$ crossing phase factors add up to $1$.
}  \label{fig:4}
\end{figure} 
\begin{proof}   
  The   decompositions of $\Gamma$ into loops are in one-to-one correspondence with the different possible pairing of the incident edges at  the vertices of $\Gamma$, carried out independently at different vertices.

 Since pairs of loops in the plane with only transversal  intersections can cross only even number of times, for each of the resulting line pattern  the  total parity factor in \eqref{eq:KW_step1} coincides with the product over the vertices of the parity of the number of line crossings at the different sites, which for future use we denote
 \be
(-1)^{n(\gamma_1,..., \gamma_n)} \ := \   \prod_j (-1)^{n(\gamma_j)} \,.  
 \ee

For vertices of degree $4$, which is the only case of relevance for $\Z^2$, the sum is over three possible connection patterns.  As is depicted in figure Fig.~\ref{fig:4}, one of them with a single crossing and the two other with with no crossing.  The net sum of the corresponding factors is
$ 1-1+1  = 1 $.
This observation completes the proof  for the Ising model  on planar graphs degree bounded by $4$.   The more general case is covered by the lemma which follows. 
\end{proof} 

\begin{lemma}  \label{lem:comb}  For the collection of pairings of  an even collection of lines in the plane which  meet transversally at a vertex of degree $2n$, the sum over the parity of the number of line  crossings  is $1$.   
\end{lemma} 
\begin{proof} The argument can be aided by a figure in which the intersection site is amplified into a disk, as depicted in Figure~\ref{fig:switch}.   

\begin{figure}[h]
 \begin{center} 
\includegraphics[width=.25\textwidth]{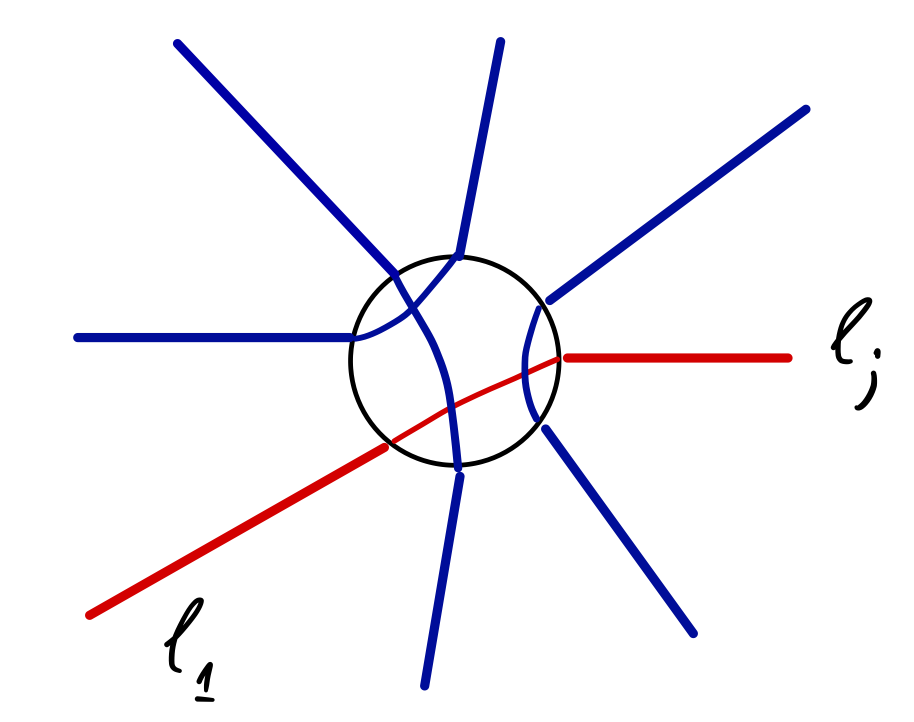}
\end{center}
\caption{An even number of lines meeting at a vertex, and one of their possible pairings, in this case of parity $(-1)^3 = -1$. }
\label{fig:switch}
\end{figure}

We prove by induction.  The case $n=1$ is trivially true.  

As the induction step, we assume the statement is true for $n$ and consider the case $n+1$.   
Let us classify the pairing configurations by  the index, $1<  j \le 2n$, of the line with which $\ell_1$ is paired.   
The line obtained by linking $\ell_1$ and $\ell_j$ splits the remaining incidental line segments  into two sets, of $(j-2, 2n - j)$ elements.   
We now note that the induction hypothesis implies that the correspondingly restricted sum of the intersection parities is $(-1)^j$.   
That is so since the parity of the other  lines' intersections with  the first one is deterministically $(-1)^j$, and the sum of the parities of the internal intersections of the rest is, by the induction hypothesis $+1$.   Therefore, the  sum in question is
\be 
\sum_{j=2}^{2n} (-1)^j \ = \ 1
\ee 
being an alternating sum of $\pm 1$ which starts and ends with $+1$.
Hence the statement exdends to $n+1$, and by induction to all $n\in \N$.   
\end{proof} 

It may be added that the line pairing can be indexed by   permutations $\pi \in S_{2n}$ such that  $\pi(2j-1) < \pi(2j)$ for  all $ j \in [1,...,n]$  
and $\pi(2j-1)$ is monotone increasing in $j$.  The line intersection parity coincides then with the permutation's parity.  The combinatorial statement of Lemma~\ref{lem:comb} is thus known in a number of forms.  \\

Lemma~\ref{lem:KW_step1} dealt with the leading collection of terms of  \eqref{eq:max_KW} corresponding to $\|\{\gamma_1,..., \gamma_n\}\|_\infty = 1 $.   The results stated before that provide effective  tools for dealing with the rest.

\begin{proof}[Proof of Theorem~\ref{thm:KW}]
Under the condition~\eqref{eq:asssmallW} on the weights $ W_e = \tanh(\beta J_e) $,  one may start from the expansion~\eqref{eq:max_KW} of $ \sqrt{ \det (1 - \mathcal{K W}) } $ into collection
of loops constrained  to have no double covered edges. 
It is natural to organize the terms of the sum in~\eqref{eq:max_KW}  according to the (disjoint) union  of edges 
$ \Gamma :=  
\sqcup_{j=1}^n  \gamma_j $.  One gets in this way:
\be \label{ad_halom}
\sqrt{ \det (1 - \mathcal{K W}) } \ 
  = \    \sum_{\Gamma \, : \, \partial \Gamma = \emptyset }   w(\Gamma) 
  \sum_{\substack{\mbox{\footnotesize loop decompositions} \\ 
   \Gamma = \sqcup \{\gamma_j\} }    } 
  (-1)^{n(\gamma_1,..., \gamma_n)} \, .
\ee
The remaining sum is the subject of  Lemma~\ref{lem:KW_step1}, which therefore implies the Kac-Ward formula.   
This proves the Kac-Ward-Feynman-Sherman theorem in case~\eqref{eq:asssmallW}. By analyticity of the determinant and the square of the (non-negative) partition 
function in the weights, the claimed relation~\eqref{eq:KW} extends to all values of $ (W_e) $  for finite graphs.
 \end{proof} 
 
 \section{The Kac-Ward formula beyond planarity} 
 
  Reviewing the proof of Theorem~\ref{thm:KW} one finds that the graph's planarity plays a role only in the last step. 
Presented below are two implications of this observation.  The first is a formula which was noted already by Sherman~\cite{She60}  and later rederived in~\cite{KLM13}.  
The second seems to be a new class of partly solvable non-planar interactions.\footnote{A  model is referred to as solvable if its partition function can be reduced to an expression of a significantly lower level of complexity than that of the partition function's defining expression.  That does not however mean that all simply stated questions have ready answers, cf. \cite{MW73}.}

 \subsection{The KW determinant on non-planar graphs}
 
The following extension of  Theorem~\ref{thm:KW} is valid for Ising spin systems with pair interactions on  graphs which need not be planar.   An example of the relevant setup is depicted in Fig.~\ref{fig_3D}.     The formula plays a role also in the analysis of strictly planar models, cf.~Section~\ref{subsec:ODpairs} and Appendix~\ref{sec:interp_per_bc}.  For clarity let us repeat that, by \eqref{Z_graph}, in the planar case the right side of~\eqref{eq:beyondP} yields the partition function.

 \begin{theorem}[Beyond planarity]     \label{thm:beyondP}  For any faithful projection of a finite graph $ \G $ (not necessarily planar) and a set of associated Ising spin-spin couplings   $\{J_{x,y}\}$  
\be    \label{eq:beyondP}
 \sqrt{\det   (1 - \mathcal{KW})}    \    
= \    \sum_{\Gamma \subset  \mathcal{E} \, : \, \partial \Gamma = \emptyset }  
  (-1)^{n_0(\Gamma)} \   w(\Gamma) 
   \ee  
 where $ n_0(\Gamma) $ is the number of $\Gamma$'s non-vertex edge crossings in $\G$'s   projection on $\R^2$. 
\end{theorem} 
 \begin{proof}
Starting from \eqref{ad_halom} it was shown that for each even subgraph $\Gamma$ the sum over loop decompositions of 
$\Gamma$ of the parity factors at each vertex adds up to $(+1)$.   
In the present context, this last statement still applies to each  vertex of the faithful projections, but not  to the non-vertex crossings.   However, at non-vertex crossings  there is no ambiguity in the local structure  of the loop decompositions.  The arguments used to conclude the proof of  Theorem~\ref{thm:KW} thus leaves one with exactly the factor $  (-1)^{n_0(\Gamma)}$, as claimed in 
\eqref{eq:beyondP}. 
\end{proof} 

\subsection{A class of partly solvable quasi-planar systems}

\begin{definition} 
In the following, we call a faithful projection of a graph on $\R^2$  \emph{quasi planar}  if each of its edges crosses at most one other edge and then at most once.  For any such projection we denote by $\mathcal E_P$ the corresponding (planar) collection of non-crossed edges.  
 \end{definition}

For convenience, in what follows we express pair interactions involving 
crossed pairs in terms of a duo of couplings $(J^+_\alpha,J^-_\alpha)$ which  correspond to a selection of a cyclic order of the endpoints of the crossed edges,   labeled as $(\alpha, j)$, $j=1,...,4$.   The distinction between $J^+$ and $J^-$ is arbitrary here, and  it is permitted for a site to be cited  in more than one $\alpha$, cf. Figure~\ref{fig_crossed}.

\begin{figure}[h]
 \begin{center} 
\includegraphics[width=.3 \textwidth]{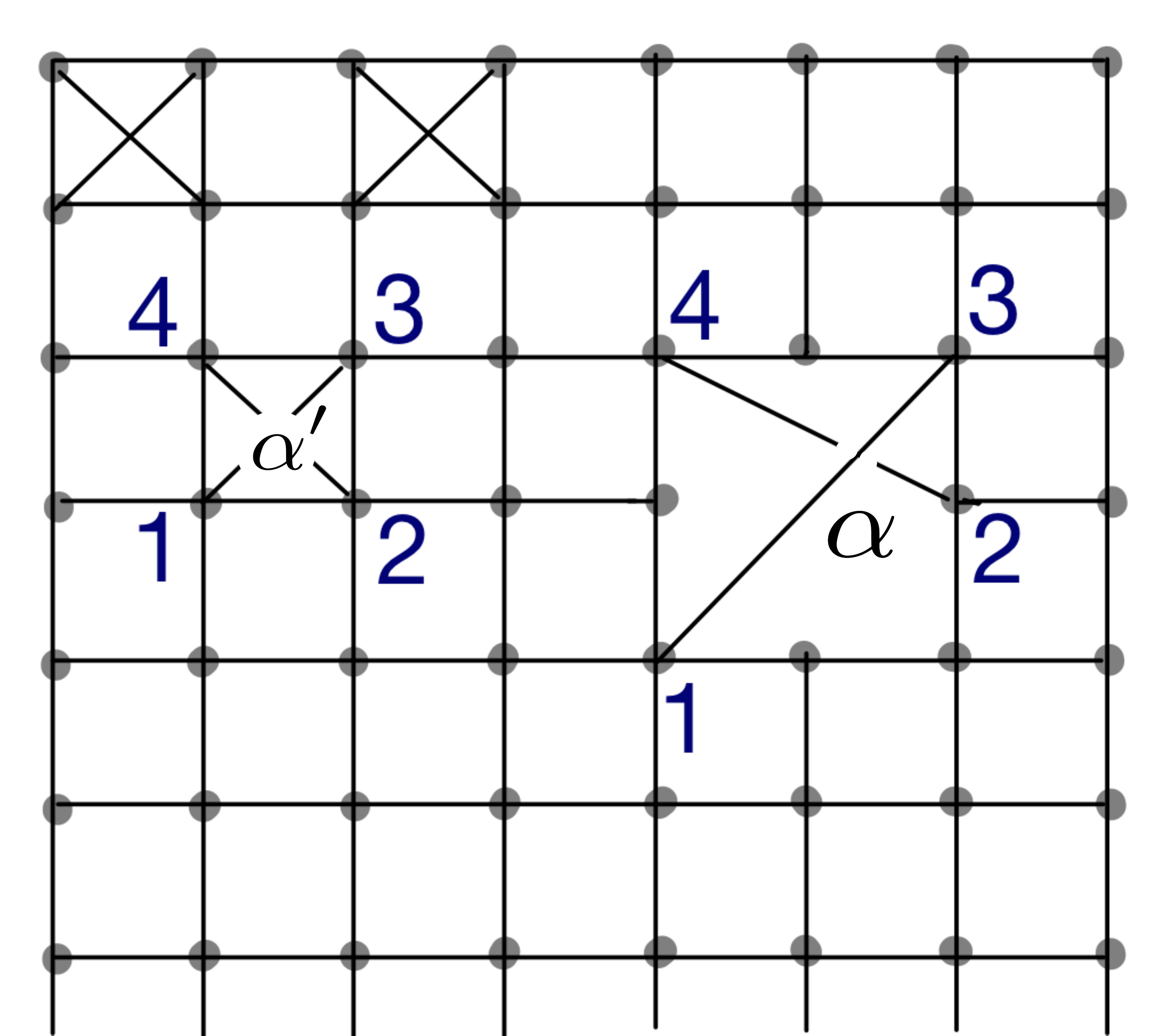}
\end{center}
\caption{A \emph{quasi planar}  graph on which the interaction of Theorem~\ref{thm_beyond_pl} is defined.  In this example there are four pairs of crossed edges, the labels of two of which are indicated explicitly.  The corresponding Kac-Ward determinant yields the partition function of such a systems with the two-body Hamiltonian modified by   four-body  ``corrector''  terms, one for each crossed pair (labeled by $\alpha$).}  
\label{fig_crossed}
\end{figure}

\begin{theorem} \label{thm_beyond_pl} Let $(\G, \mathcal E)$ be a finite graph with a faithful quasi planar projection, whose collection of internal intersection points is denoted by $S$ and labeled by $\alpha$, and a pair-interaction Hamiltonian presented as
\be 
H(\sigma) \= -  \sum_{(u,v)\in \mathcal E_P} J_{u,v} \sigma_u \sigma_v  
\ -\  \sum_{\alpha \in S}  \left[ J^+_\alpha \sigma_{\alpha,1} \, \sigma_{\alpha,3}  \ + \  J^-_\alpha \sigma_{\alpha,2} \, \sigma_{\alpha,4} 
\right] 
\ee  
where the first sum is over the uncrossed edges, and  the second  over pairs of  crossed edges, indexed by $\alpha$,  using the convention explained above. 

Then for each temperature $\beta^{-1}$ a determinantal solution exists for the temperature-dependent  Hamiltonian
\be
H^{corr}_\beta(\sigma) \ := \  H(\sigma) + V_\beta(\sigma) \, , 
\ee 
with the ``corrector'' term consisting of the four-spin  interaction 
\be
V_\beta(\sigma)  \=     \sum_{\alpha \in S}
\displaystyle 
R_{\alpha} \, \prod_{j=1}^4 \sigma_{\alpha,j} 
\ee 
at strengths $ R_{\alpha}\in \R$  given by the relation:
\be \label{R_solution}
\tanh(\beta R_\alpha)  \ = 
 \frac{Y_\alpha}{ 1+ \sqrt{1-Y_\alpha^{2}}}
\ee 
where 
\be \label{Y}
Y_\alpha \= \tanh(2\beta J^+_\alpha) \tanh(2 \beta J^-_\alpha)\,.
\ee
\end{theorem} 

An explicit expression for the partition function of the modified Hamiltonian is stated at the end of the proof.  
As is the case for the Onsager and Kac-Ward formulae, the statement is not restricted to ferromagnetic couplings.
The non-planar additions to the hamiltonian can be made  preserving translation invariance, or in a disordered fashion.   

\begin{proof}  In the following argument we  abbreviate by   $ \t X$  the $\beta $-dependent quantities 
\be \t X = \tanh (\beta X) \, , \quad \mbox{for  $X = J^\pm_\alpha,  R_\alpha$\,.} 
\ee 

For each $\beta$ (and $h=0$) the partition function of $H^{corr}_\beta$ is (at $h=0$)  
\be 
Z^{corr}(\beta)  \  := \    \widetilde  Z^{corr}(\beta)  \times  \left[ 2^{|\mathcal{V}|}  \prod_{\{x,y\}\subset  \mathcal{E}_P} \cosh (\beta J_{x,y} )\right ]   \, \prod_{\alpha \in S} 
\left[ \cosh(\beta J^+_\alpha) \cosh(\beta J^-_\alpha)  \cosh(\beta R_\alpha) \right] \,  , 
\ee
with 
\begin{multline} 
 \widetilde  Z^{corr}(\beta)  \= \frac{1}{ 2^{|\mathcal{V}|} } \sum_\sigma
 \prod_{\{x,y\}\subset  \mathcal{E}_P} \left[  1 + \t{J}_{x,y} \sigma_x \sigma_y
 \right]  \, \times \\ 
 \prod_\alpha \left[(1+ \t{J}^+_\alpha \sigma_{\alpha,1} \, \sigma_{\alpha,3} )  
(1+ \t{J}^-_\alpha \sigma_{\alpha,2} \, \sigma_{\alpha,4} )   
\Big(1- \t{R}_\alpha \prod_{j=1}^4\sigma_{\alpha,j} \Big)  
 \right] 
 \end{multline}  

Denoting $N_\alpha   = 1- \t{J}^+_\alpha \t{J}^-_\alpha \t{R} $, \, 
the factor corresponding to each $\alpha$ expands to:
\ba 
N_\alpha \, 
\left(1+ A_\alpha \, \sigma_{\alpha,1} \, \sigma_{\alpha,3}   
+ B_\alpha\,  \sigma_{\alpha,2} \, \sigma_{\alpha,4}     
+ C_\alpha \prod_{j=1}^4\sigma_{\alpha,j}  
 \right)  
\ea 
with
\be \label{corr_exp}
N_\alpha A_\alpha \= \   \t{J}^+_\alpha - \t{J}^-_\alpha \t{R}_\alpha  \,,   \quad  N_\alpha B_\alpha \=  \t{J}^-_\alpha - \t{J}^+_\alpha \t{R}_\alpha  \,,   \quad  
 \quad N_\alpha C_\alpha \= \   \t{J}^+_\alpha  \t{J}^-_\alpha - \t{R}_\alpha \,.  
\ee 
Our goal is to find $R_\alpha$ such that  
\be  \label{AHA}
C_\alpha \= - A_\alpha \cdot B_\alpha  \, .   
\ee
\mbox{}\\ 

The condition 
\eqref{AHA}  
yields a quadratic equation for $\t{R}_\alpha$ in terms of $\t{J}_\alpha^\pm$ which  through the identity 
 $(\t{J} +\t{J}^{-1})/2 = 1/\tanh(2\beta J)$ reduces to
\be \label{R2}
\t{R}_\alpha^2 - \frac{2}{Y_\alpha} \, \t{R}_\alpha + 1 \= 0  \,. 
\ee 
For any set of the pair coupling $(J^-_\alpha, J_\alpha^+)$ one has 
$|Y_\alpha| <1$ (cf. \eqref{Y}), in which case  \eqref{R2} 
has two real and unequal solutions whose   product is $1$.  
The solution satisfying  $|\t{R}_\alpha| <1$ (a necessary condition for $\t{R}$ to coincide with $\tanh(\beta R)$ at real $R$) is given by \eqref{R_solution}.
 
Under the key condition \eqref{AHA} (which reverses the sign with which $A_\alpha B_\alpha$ appears in the planar case)  the graphical expansion of  $ \t{Z}^{corr}(\beta)$ along the lines of Lemma~\ref{lem:evenGraph} coincides with the expression found on the right of \eqref{eq:beyondP} for the Kac-Ward determinant of a modified matrix -- multiplied by  $\prod_\alpha N_\alpha$.   The modified matrix $\mathcal{W}^{corr}$ is obtained from $\mathcal{W}$ which corresponds to the initial pair interaction in $H(\sigma)$ by changing the weights of  the crossed edges  so that for each $\alpha\in S$
\begin{eqnarray} 
\t{J}^+_\alpha  &\mbox{is replaced by} \quad
\left(\t{J}^+_\alpha -   \t{J}^-_\alpha R_\alpha\right) \big/  N_\alpha \notag \\   
 \t{J}^-_\alpha  &\mbox{is replaced by}  \quad 
\left( \t{J}^-_\alpha -   \t{J}^+_\alpha R_\alpha \right) \big/  N_\alpha
\end{eqnarray}  

The net result of the above considerations is 
\be
\widetilde  Z^{corr}(\beta) \=  \sqrt{\det \left(1 - \mathcal{K} \mathcal{W}^{corr} \right) }  \, \prod_{\alpha\in S} N_\alpha
\ee  
where $\mathcal{K}$ is the Kac-Ward matrix \eqref{eq:defK} corresponding to the graph's given planar projection (indexed by the set of the graph's oriented edges).  

The main statement can be rephrased by saying that   \eqref{AHA}  defines a surface in the model's parameter space (among which we include the temperature) along which the enhanced quasi planar model is solvable by the extended Kac-Ward formula.  
 \end{proof}

\section{Emergent fermionic structure} \label{sec:KadCev}

Onsager's solution was algebraic in nature.  Its different elaborations and simplification show 
the utility  in this context of a number of anti-commuting structures, which show in various forms.

 \begin{enumerate} 
 
 \item The spinor-based derivation of Onsager's solution by  B. Kaufman~\cite{Kauf49}
  \item The determinantal representation~\eqref{eq:KW} of the partition functions by Kac and Ward~\cite{KW52}  and later, in the context of a dimer partition function, by Kasteleyn~\cite{Kas63}.  
 
 \item The Schultz, Mattis and Lieb \cite{SML64} representation of the transfer matrix as the Hamiltonian of non-interaction fermions.
 
 \item The Kadanoff-Ceva~\cite{KC71}  fermionic spinors constructed as linked pairs of order-disorder operators, whose definition is restated below.  Their educated guess as to  the power laws which may  be associated with the resulting spinor structures yields  considerable insight on the values and structure   of the model's critical exponents. 

\item The Pfaffian structure of the boundary correlation functions of the model on planar graphs, which was pointed out by Groeneveld, Boel and Kasteleyn~\cite{GBK77}.   
 
 \item Smirnov's parafermionic amplitude~\cite{Smi10}, which  plays an important role in his derivation of the scaling limit of the nearest-neighbor Ising model on  $ \mathbb{Z}^2 $.   Lis~\cite{Lis14} and Dubedat~\cite{Dubedat} have linked the related kernel to the inverse of the Kac-Ward matrix.   
 The theory's extension to the broader context of planar graphs was offered in the work of Chelkak and  Smirnov \cite{CS12}, cf.~\cite{CCK16} for a recent review and analysis.        
\end{enumerate} 

Let us add that the last two fermionic structures are not identical.  
Smirnov's parafermionic amplitude plays a role  in  the broader class of $Q$ state Potts models.  In comparison  the more limited Pfaffian relations suggests the presence, in planar graphs, of  \emph{non-interacting} fermionic degrees of freedom.     
(A similar structure can be found in the theory of planar graphs' dimer covers~\cite{ALW}.)   
While the Pfaffian relation  of regular correlations does not extend to the bulk, it does find its extension in the $n$-point correlation functions of the order-disorder operators, as is shown next.

\subsection{Disorder operators} 

The  Kadanoff and Ceva~\cite{KC71}  \emph{disorder operators} are associated with vertex-avoiding  lines,  drawn over the graph's planar projection.  Associated to each line~$\ell$ is a transformation of the Hamiltonian, $H \mapsto R_\ell H$,   which flips the sign of the couplings $(J_{x,y})$ for edges $\{x,y\} \in \mathcal {E}$ which are crossed an odd number of times by~$\ell$.   

Extending the notion of correlation function from the mean values of spin products,   
$ \langle \prod_j \sigma_{x_j}   \rangle_{\beta, 0}^\G  $, to the effect of more general operations on functional integrals or sums, one may define: 
\be 
\langle \prod_j R_{\ell_j}   \rangle^{\G} \ : = \  
\frac{\sum_{\sigma} e^{-\beta (R_{\ell_n} \circ \dots \circ R_{\ell_1} H)(\sigma)  }} 
{  Z_{\G}( \beta,0) }  \, .
\ee  
where we subsequently assume that $ h =0 $.
Under the Kramers-Wannier duality, 
for any planar graph and collection of lines $\{\ell_j\}$
 \be  \label{disorder} 
\langle \prod_j R_{\ell_j}   \rangle^{\G} \ =  
\langle  \prod_j \prod_{y\in   \partial \ell_j  } \sigma_{y}    
\rangle^{\G^*} 
\ee
where the expectation value $ \langle \cdot \rangle^{\G^*} $  refers to an Ising system on the dual graph $ \G^* $, at the  dual  values of $K_e \equiv \beta J_e$.   
The duality  is the  $Q=2$ case of the $Q$ state Potts models  relation 
\be 
\left( e^{2K _e} -1  \right) \,  \left( e^{2K^* _{e^*}} -1  \right) \ = \  Q  
\,  ,  
\ee 
which in terms of the weights $W_e=\tanh K_e$ 
can be restated as: 
\be 
W^*_{e^*} \ = \  \frac{1-W_e }{1+W_e}  \, .
\ee

Long-range order in the dual model indicates (in the homogeneous case) that the original spin system is in its disordered phase.  This, and the relation \eqref{disorder} led to the proposal to call $R_\ell$ \emph{disorder operators}~\cite{KC71}.   Adapting a  convention for a canonical  choice of the line trajectory given its end points,  leads to a useful, though slightly deceptive, picture of the disorder variables as local operators, whose correlation functions coincide with the spin correlation functions of the dual model.   
\\ 
 
The disorder variables can also be viewed as  partial implementations of a gauge symmetry. This gauge symmetry perspective  is useful in extending the construction to other systems, e.g. dimer covers~\cite{ALW}.
For Ising systems at $h=0$ any preselected spin flip can be viewed as gauge transformation, under which the physics is preserved  but the interaction appears modified.  Due to this symmetry, 
$ 
\langle \prod_j R_{\ell_j}   \rangle  =1 $ for any collection of loops, i.e., $\partial \ell_j = \emptyset$.   For open ended lines, the Gibbs equilibrium system changes in an essential way, but modulo gauge transformation it is a homotopy invariant function of the line $\ell$, i.e., up to spin flips  it depends only on $\partial \ell$.     
 More explicitly, the choice of paths by which specified boundary points of $\ell_j$ are linked do not affect  the expectation values in \eqref{disorder}.   That is not quite true for the correlation functions of mixed operators, such as:
\be 
\langle \prod_k \sigma_{x_k} 
\prod_j R_{\ell_j}   \rangle ^{\G} \ : = \  
\frac{\sum_{\sigma} \prod_k \sigma_{x_k}  e^{-\beta (\prod_ j R_{\ell_j} H_\mathcal J)(\sigma)  }} 
{  Z_{\G}( \beta,0) }  \, .
\ee  
However, these just change sign whenever one of the lines is deformed over one of the spin sites, while the lines' end points are kept fixed.

\subsection{Fermionic order-disorder operator pairs}  \label{subsec:ODpairs}

Pairs of order and disorder variables exhibit an interesting emergent structure which,   as was pointed out in~\cite{KC71}, gives an explicit expression to Kaufman's  spinor operators.  The presence of spinor algebra was recognized early on in her  algebraic approach to the model's solution~\cite{Kauf49}.  Since then it has been realized to be also of great relevance for the field theoretic description of the  model's  critical state.   

\begin{figure}  \begin{center}
\includegraphics[width=0.4 \textwidth] {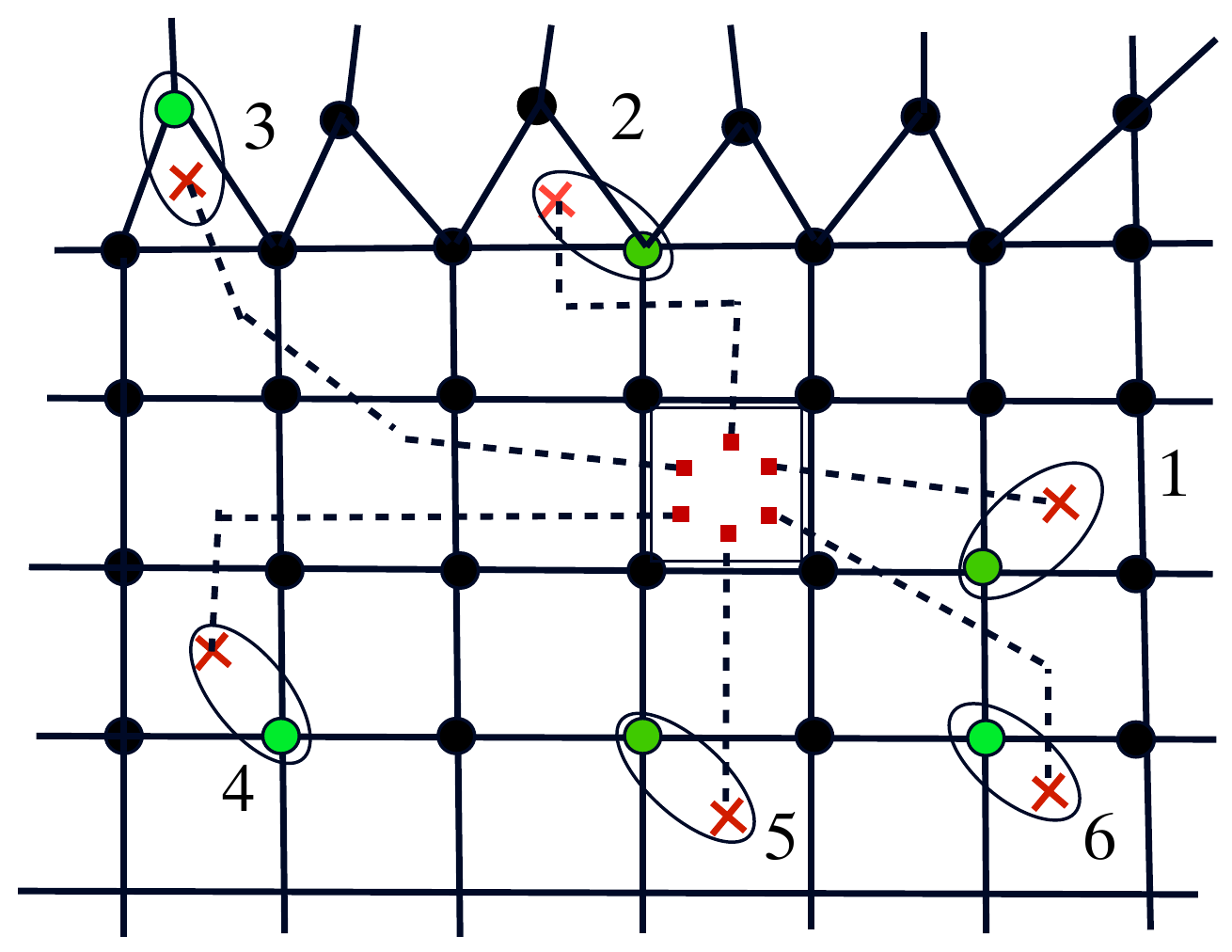} 
\caption{Order-disorder variables for a planar graph.  The disorder variables $\tau_{\ell_j}$ are defined through a  set of lines $\ell_j $, each linking a dual site $x_j^*\in \G^*$ which is a neighbor of $x_j$ in $\G \times \G^*$  with a common dual site $x_0^*\in \G^*$, called grand central. The disorder lines $ \ell_1, \ell_2, \dots $ are enumerated cyclicly in the order of the lines' emergence from the grand central $x_0^* $.    (The figure is reprinted from \cite{ALW} where  it was  used to demonstrate a related structure in dimer covers.)}
\label{fig:ord_disord}  
\end{center}
\end{figure} 
 \begin{definition}
Let  $\G=(\mathcal{V},\mathcal{E})$ be a planar graph, and $ \mathcal W: \mathcal E \mapsto \C$ a set of  edge weights ($\mathcal W_{x,y} = \tanh(\beta J_{x,y})$).   Pairs of  \emph{order-disorder}  operators  consist of  products
\be 
\mu(x_j,\ell_j) \ = \ \sigma _{x_j} R_{\ell_j}  
\ee 
where $ \sigma _{x_j}$ are Ising spins associated with sites $x_j\in \mathcal{V}$, and $R_{\ell_j} $ are the disorder operators associated with  piecewise differentiable lines in the plane, which avoid the sites of $\G$,  such that one of the end points $x_j^*$ of each $\ell_j$ is in a plaquette of $\G$ whose boundary includes $x_j$, and the other is in a common plaquette of $\G$  to which we refer as the \emph{grand central} $ x_0^*$. 
\end{definition}

The defining feature of fermionic variables is the change of sign  of their correlation functions, including with other variables, which occurs when pairs of fermionic variables are exchanged.   To see this occurring in the present context one should note that   the disorder operators $\mu_j$ are not quite local:  tethered to each  is a disorder line.   When one of these lines is made to pass over a spin site the product $\prod_j \mu_j$ changes sign.   However, the correlation function $\langle \prod_j \mu_j \rangle $ depends on the lines $\{\ell_j\}$ only through the combination of their end points and a collection of   binary  variables $\tau(\ell_j ,\sigma_{x_k}) = \pm 1$  which change  sign upon the passing of a line over a spin site.   

To organize the information in terms of a minimal set of variables it is convenient to adapt a convention, 
an example of which is indicated in Fig.~\ref{fig:canonical},
 for a `canonical assignment' of a disorder line $\ell_j$ to each potential disorder site $x^*_j$.   The correlation is then a function of  the configuration of the order-disorder sites $(x_j,x^*_j)$, multiplied by the  binary variable $\tau = \prod_{j,k} \tau_{j,k}$.  The latter changes sign whenever one of the disorder sites moves across a well defined set of lines in the plane in which the graph is embedded.   Equivalently, one may regard the correlation function as a single-valued function defined over a  two-sheeted cover of the graph.  This perspective is developed further in~\cite{CS12}.

In a sense similar to the above, a sign change occurs also when any of the order-disorder pair is  rotated by $2\pi$ (e.g. when the disorder site is taken around its associated spin site).   In this  the order-disorder operator pairs  resemble \emph{spinors}.   As was pointed out in~\cite{KC71},  that provides a strong hint that  in the scaling limit of the critical model the corresponding operator is of conformal dimension $1/2$.  Exact methods confirm that to be the case.   \\ 

 \begin{figure}  \begin{center}
\includegraphics[width=0.3\textwidth] {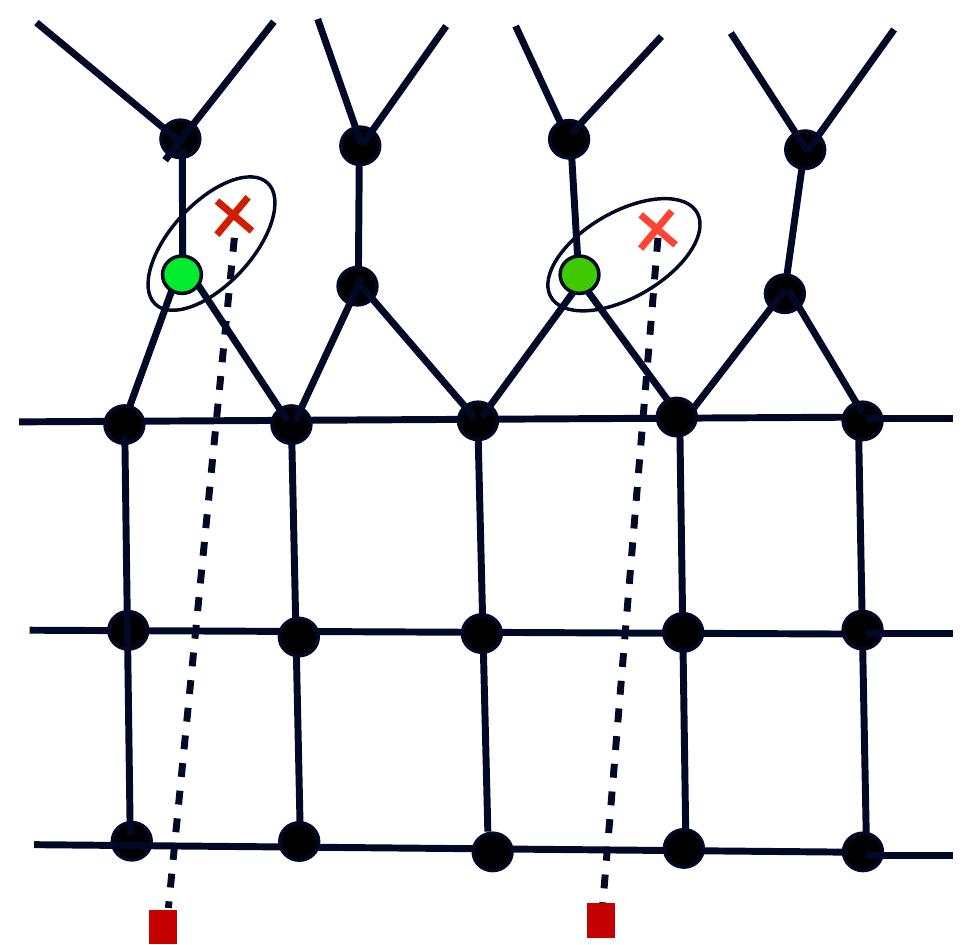} 
\caption{One of the choices for  association of disorder lines with disorder sites.  Under any such convention the product of disorder variables may be presented as associated with just disorder sites, and a  single $\pm1$ variable which changes sign whenever one of the disorder lines passes over a spin site.  }
\label{fig:canonical}  
\end{center}
\end{figure}

In discussing  pair correlations $ \mu_j = \mu(x_j,\ell_j)$, $ j = 1,2 $, it is  natural to replace the two lines by a single line $\ell_{1,2}$ with end points 
$\partial \ell_{1,2} = \{x_1, x_2\}$.  In the setup described above, this is achieved by linking the pairs of lines $(\ell_1, \ell_2)$  within  $ x_0^*$, possibly deforming the resulting curve to fit a convenient convention, and linking each endpoint $ x_j^* $ through a straight line to $ x_j $.   
In case of just two order-disorder pairs,  a natural convention is to link $\{x^*_1, x^*_2\}$ by a straight line adjusted so that the constructed lines would avoid the vertices of $\G$. Associated to
$  \ell_{1,2} $ is an angle
\be
\theta_{ 1,2} := \Delta \mbox{\rm arg}(\ell_{1,2}) \in (-\pi,\pi] 
\ee
which describes the change of the tangent's argument when starting from $ x_2$ and following along $   \ell_{1,2} $ to $ x_1 $ thus explicitly including the turn from the straight line connecting 
$ x_2 $ to $ x_2^* $ to the line connecting $ x_2^* $ and $ x_1^* $ and likewise the turn from that line to the straight line connecting $ x_1^* $ and $ x_1 $.

The main new result in this section is that the corresponding two point function is given by the Kac-Ward matrices' resolvent kernel
\be  \label{def:G} 
G(e_1, e_2) \ := \    
  \left( \frac{1} {1 - \mathcal {K W} }\right)_{e_1,e_2}  \, , \quad e_1,e_2 \in \widehat{\mathcal{E}} \, . 
\ee 
More explicitly:  

  \begin{theorem}\label{thm:Green}  For an Ising model on a finite planar graph $ \G $ with pair interactions whose sign can be arbitrary and any pair of order-disorder variables:
\begin{align}\label{eq:orddisG}
 \langle \mu_1 \mu_2 \rangle  & \equiv \langle \sigma_{x_1}  R_{\ell_1}   \sigma_{x_2}  R_{\ell_2}  \rangle  \notag \\
& = \ -  \Re\Big[ e^{ \tfrac{i}{2} \, \theta_{1,2}} \,
\sum_{\substack{e_1: o(e_1) = x_1 \\   e_2: t(e_2) = x_2}}   e^{\tfrac{i}{2} \angle(e_1,\mu_1)}\, 
W_{e_1} \, G(e_1, e_2) \, e^{- \tfrac{i}{2} \angle(e_2,\mu_2)}\Big]\, . 
\end{align}
The summation extends over oriented edges $ e_1,e_2 \in \widehat{\mathcal{E}} $ with origin $o(e_1) = x_1  $ and terminal point $ t(e_2) = x_2 $. The angles  $ \angle(e_j,\mu_j)  $ with $ j = 1,2 $ are the change of the tangent's argument from the straight line $ (x_j^*,x_j) $ 
to the origin of  $e_j$. 

 \end{theorem} 
\begin{proof}
Let $\widetilde \G$ be the graph obtained by adding to the edge set of 
 $\G$ the  piecewise differentiable line $\ell$ with $\partial \ell  =\{x_1, x_2\} $ which includes $ \ell_{1,2} $ as described above.
 Although  $\widetilde \G$ is not a planar graph,    
Theorem~\ref{thm:beyondP} applies to it, and implies that  
\begin{align} \label{55} 
 \sqrt {\det
\left( 1 -  \widetilde {\mathcal K} \widetilde {\mathcal W} \right) }  =  \sum_{\Gamma \in \widetilde{\mathcal{E}} \, : \, \partial \Gamma = \emptyset }  
  (-1)^{n_0(\Gamma)} \   \widetilde w(\Gamma) \, , & \\     \prod_{\{x,y\} \in \Gamma}    W_{\{x,y\} } & = \widetilde w(\Gamma)  \notag
\end{align}
 with $ \widetilde {\mathcal K} \widetilde {\mathcal W}  $ the natural extension of the Kac-Ward matrix $ \mathcal{K W} $ to the oriented edges of the augmented 
 graph $\widetilde \G$ with weight $ W_\ell $ on the added line $ \ell $ and its reverse $ \overline{\ell} $.  
Orienting $ \ell = (x_2,x_1) $ in the direction towards $ x_1 $, the angles arising in the definition of the matrix~\eqref{eq:defK} are
\be
 \widetilde{\mathcal K}_{\ell,e_1}  = \exp\left[ \tfrac{i}{2}\left( \angle(e_1,\mu_1)+   \theta_{1,2}  \right) \right]  \, , \qquad 
 \widetilde{\mathcal K}_{e_2,\ell}   = \exp\left[ - \tfrac{i}{2} \angle(e_2,\mu_2) \right] \, ,
\ee
and likewise for $ \overline{\ell} $.

For a given $ \Gamma \subset \widetilde{\mathcal{E}}$ with $ \ell \in  \Gamma $, the number  $ n_0(\Gamma) $ in the right side of~\eqref{55} equals the number of crossings of edges $ \mathcal{E} \subset \Gamma$ with the extra edge $ \ell $. 
 It is natural to split the sum in the right side into terms for which $  \ell  \not\in  \Gamma $ and $  \overline\ell  \not\in  \Gamma$ and the remaining terms.  According to~\eqref{Z_graph} the former sum up to $Z_\G(\beta,0) $. The latter contribution is linear in $ W_\ell $. Taking logarithmic derivatives  and evaluating at $ W_\ell = 0 $ we arrive at
 \begin{align}
 \frac{d}{d W_\ell} \ln \rm{LHS}  \, \eqref{55} \, \big|_{W_\ell = 0 } & = \frac{1}{ Z_\G(\beta,0) } \, \sum_{\substack{\Gamma \subset \mathcal{E} \\ \partial \Gamma = \{ x_1,x_2 \} }} (-1)^{n_0(\Gamma)} \, w(\Gamma)  = \langle \mu_1 \mu_2 \rangle  \, .
 \end{align}
  The last equality is due to~\eqref{eq:spincorr} with the factor $ (-1)^{n_0(\Gamma)} $ taking into account the reversal of the Ising couplings' signs on edges $ \mathcal{E} $ which are crossed by the disorder line $ \gamma_{1,2} $.

Proceeding as in~\eqref{eq:diffdet} the logarithmic derivative of the determinant with respect to the weight of the  added edge (and all other couplings  initially small enough) gives
\begin{align}\label{eq:77}
  & \frac{d}{d W_\ell} \ln \rm{RHS}  \, \eqref{55}  \, \big|_{W_\ell = 0 } \notag \\ & = - \frac{1}{2} \tr\left( \frac{1}{1-  \widetilde {\mathcal K} \widetilde {\mathcal W} }  \, \widetilde{\mathcal K} \, P^{(\ell)} \right) \, \Big|_{W_\ell = 0 }  = -  \frac{1}{2} \sum_{n=1}^\infty \tr  \left( \left(\widetilde{\mathcal K} \widetilde{\mathcal W} \right)^n \widetilde{\mathcal K} P^{(\ell)}  \right)  \, \Big|_{W_\ell = 0 } \notag \\
&=  -  \frac{1}{2} \sum_{e_1,e_2 \in \widehat{\mathcal{E}} } \left( \widetilde{\mathcal K}_{\ell,e_1} W_{e_1}  G(e_1,e_2) \, \widetilde{\mathcal K}_{e_2,\ell} +  \widetilde{\mathcal K}_{\overline{\ell},\overline{e}_2} W_{e_2}  G(\overline{e}_2,\overline{e}_1) \,  \widetilde{\mathcal K}_{\overline{e}_1,\overline{\ell}} \right) \, . 
 \end{align}
The second equality is a consequence of the resolvent expansion, which is 
valid for small enough couplings, cf.~\eqref{eq:asssmallW}. It facilitates the proof of the last line in which  $ W_\ell $ was taken to zero. The summation in the last line is non-zero unless  $ o(e_1) = x_1 $ and $ t(e_2) = x_2 $. In this case, we have
$$
\widetilde{\mathcal K}_{\overline{\ell},\overline{e}_2}\,   \widetilde{\mathcal K}_{\overline{e}_1,\overline{\ell}} =  \exp\left[ - \tfrac{i}{2}\left( \angle(e_1,\mu_1)-  \angle(e_2,\mu_2)+   \theta_{1,2}  \right) \right] \, . 
$$
Since also $ W_{e_2} \, G(\overline{e}_2,\overline{e}_1)  = W_{e_1} \, \overline{ G(e_1,e_2)} $, the sum in the right side of~\eqref{eq:77} is the real part of its first term and hence agrees with the right side of~\eqref{eq:orddisG}. 

Since both sides of~\eqref{eq:orddisG} are meromorphic functions of the couplings $ (W_e) $,  the validity of this equality extends by analyticity  from initially small to all values of the pair interactions.
\end{proof} 

\subsection{Linear equations for the mixed correlation function} 

A characteristic  property of the resolvent kernel of a local operator is that it satisfies a linear relation, and has a natural path expansion.   
Correspondingly, the resolvent of the Kac-Ward matrix, which was 
introduced in \eqref{def:G} satisfies the  linear relation 

\be \label{eq:Dirac} 
G(e_1,e_2) = \sum_{e'} (\mathcal {KW})_{e_1,e'} \, G(e',e_2) \ + \  \delta_{e_1,e_2}\, ,
\ee  
which has some  resemblance to the Dirac equation.  It  is reminiscent also of the one satisfied by  Smirnov's para\-fermionic observable, which  played an important role in the analysis of the critical model's scaling limit~\cite{CS12}. In fact, as has been pointed out by Lis~\cite{Lis14}, the Green function is equal to the complex conjugate of the Fermionic generating functions of~\cite{CS12}.  \\   

 \begin{figure}  \begin{center}
\includegraphics[width=0.3\textwidth] {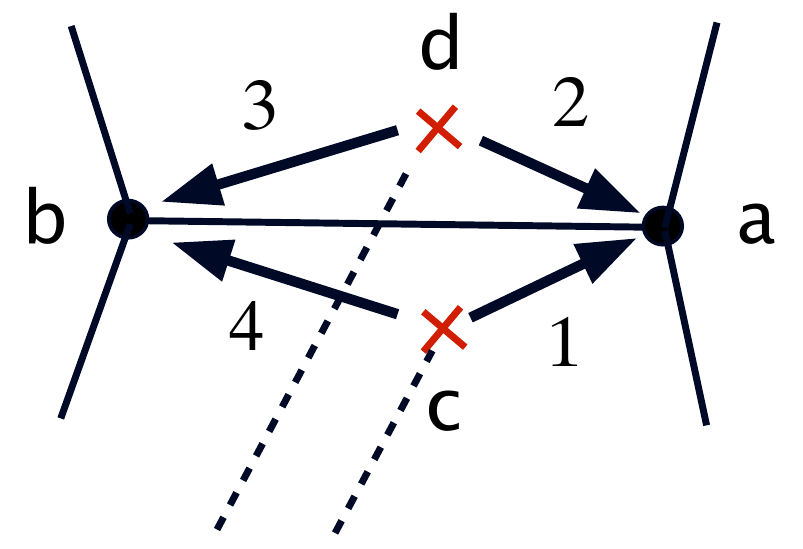} 
\caption{An edge with four possible disorder variables attached to it. The dashed lines indicate the location of the disorder lines.}
\label{fig:Dotsenko}  
\end{center}
\end{figure} 
For small enough weights $(W_e)$ (cf.~\eqref{eq:asssmallW}), the resolvent's definition may be expanded into a convergent path expansion: 
\be  G(e_1, e_2)   =   \left(\frac{1} {1 - \mathcal {K W} } \right)_{e_1,e_2}   =   
\sum_{ \widehat \gamma: e_1 \to e_2}  e^{\, i\Delta \text{arg}( \widehat \gamma)
  /2}  \, \prod_{j}  W_{ \widehat \gamma_j} 
\ee
Inserting this path expansion into~\eqref{eq:orddisG} yields a path expansion for the
 order-disorder correlators  $  \langle \mu_1 \mu_2 \rangle $. Such path expansions are useful in the derivation of other identities. One example are the linear relations initially derived by Dotsenko~\cite{Dot83}. They concern the four correlations $ \chi_j \equiv  \langle \mu_{j} \mu_0 \rangle $, $ j \in \{ 1, \dots, 4 \} $, obtained by fixing $ \mu_0 $
together with an edge $ e =(a,b) $ and its two adjacent plaquettes in each of which we pick one point $ c $ and $d $. In case the disorder 
lines attached to $c $, $ d $ are as depicted in Figure~\ref{fig:Dotsenko}, one has
\begin{align}\label{eq:Dotsenko}
	W_e \, \left( \chi_1 + \chi_4 \right) \ & = \   \left( \chi_2 - \chi_3 \right) \notag \\
	W_e \, \left( \chi_2 + \chi_3 \right) \ & = \   \left( \chi_1 - \chi_4 \right) \, . 
\end{align}
In  $ \chi_1 + \chi_4 $, the contribution of paths which exit at $ a$ towards the edge $ e $ cancel due to the angles's difference of $ 2 \pi $ when flowing into $ e $ along the respective disorder lines.
All edges  other than $ e $ which originate in $a  $ have an identical contribution in $ \chi_1 $ and $ \chi_4 $. The difference $  \chi_2 - \chi_3 $  on the other hand cancels all paths
emerging from $ b$ in directions other than $ \overline{e} $ on which the contributions from $ \chi_2 $ and $ - \chi_3 $ coincide. Due to the non-backtracking property, the paths which emerge from  $ \overline{e} $ at $ a $  contribute equally to  $ \chi_1 + \chi_4 $. The respective accumulated angles in these paths in $  \chi_1 $ and $  \chi_2 $
also agree. This proves the first relation~\eqref{eq:Dotsenko}. The other one is proven similarly.
 
It may be added that the fact that certain combinations of Ising model correlation functions obey simply stated
relations has been noted in a number of works albeit through rather different means, cf. \cite{WPW81,Pal07,CCK16} and references therein.  

\subsection{Pfaffian structure of correlations}

Let us mention in passing that the combinatorial methods presented here allow to prove also that 
in planar graphs the correlation functions of the order-disorder are not only fermionic (changing sign upon  exchange, as explained above) but, furthermore, resemble the correlation functions of \emph{non-interacting} fermions.  
This  feature of the model was noted and put to use  in the early algebraic approach  of Schultz, Mattis and Lieb~ \cite{SML64}, and numerous works which followed, in the context of a transfer matrix approach.  However it is valid also for arbitrary planar graphs. 

\begin{definition} A collection of order-disorder variable pairs is said to be  \emph{cyclicly ordered} if their lines do not cross, and the pairs   are labeled in a cyclic order  relative to $\ell_j$'s intersections with  the edge boundary of $ x_0^*$. 
\end{definition}

\begin{theorem}[Pfaffian correlations] \label{thm:Pf_OD}  
For a finite planar graph $\G$ with edge weights $ K \,: \mathcal E \mapsto \C$, for any collection of canonical pairs of  
order-disorder variables $ \mu_j=\mu (x_j,\ell_j) $,  $ j \in \{ 1, \dots , 2n\} $, ordered cyclicly relative to the grand central 
\be  \label{eq:Pf_gen} 
 \langle  \prod_{j=1}^{2n} \mu_j  \rangle
\ = \  \sum_{ \pi \in \Pi_n} \sgn(\pi) \, \prod_{j=1}^n  
\langle \mu_{ \pi(2j-1)} \,  \mu_{ \pi(2j)}  \rangle   \ \equiv \ \Pf_n\left( \langle \mu_j \mu_k   \rangle  \right)  \, ,
\ee
where $ \Pi_n $ is the collection of pairings of $ \{1, \dots, 2n \} $ and $ \sgn(\pi) $ is the pairing's signature.
\end{theorem} 

In case the spin variables lie along a connected boundary segment, the  lines $\ell_j$ can be chosen so they do not cross any edge,  and the operators $\tau_{\ell_j} $ act as identity.
Consequently, the order-disorder pairs reduce to  regular spin operators.  In this case the  $n$-point  boundary correlation functions have the Pfaffian structure, as was first noted by Groeneveld,  Boel, and Kasteleyn  \cite{GBK77}. 

In case the monomers  $ \{ x_{2j-1},x_{2j} \}$ are pairwise adjacent, 
the disorder lines may be chosen so that their actions are pairwise equivalent, and thus cancel each other.  In that case the pairwise product of two order-disorder variables reduces to a  an ordinary product of monomers, i.e., a dimer $ 
\mu_{2j-1} \mu_{2j}  =  \sigma_{x_{2j-1}} \sigma_{x_{2j}}$, so that  
\be
 \big\langle  \prod_{j=1}^{2n} \mu_j  \big\rangle =  \big\langle  \prod_{j=1}^{n}\sigma_{x_{2j-1}} \sigma_{x_{2j}} \big\rangle \, .
\ee
Related observations were made and applied in the discussion of critical exponents by Kadanoff~\cite{Kad66,KC71}.

 We omit here the proof of Theorem~\ref{thm:Pf_OD}, since the statement is more thoroughly  discussed in~\cite{ADTW17}.  Furthermore, 
a  derivation by combinatorial path methods  of a closely related dimer-cover version of Theorem~\ref{thm:Pf_OD}  was recently presented in \cite{ALW}, with which Fig.~\ref{fig:ord_disord} is shared.  Its argument applies verbatim  here as well. 
Let us add that related statements can also be found in both the classical text of the Ising model~\cite{MW73} and in the review of more recent developments by  Chelkak, Cimasoni and Kassel \cite{CCK16}.

\appendix

\section{From the Kac-Ward formula to Onsager's explicit solution}  

Following is a summary of the calculation  by which  Onsager's free energy  \eqref{eq:Onsager} emerges from the Kac-Ward determinantal formula \eqref{eq:KW}, which we repeat here: 
\be  \label{KW_det}
 Z_\G(\beta,0) =  \left( 2^{|\mathcal{V}|} \prod_{\{x,y\}\subset  \mathcal{E}} \cosh (\beta J_{x,y}) \right)  \sqrt{\det (1 - \mathcal{K W} )}
\ee
with $ \mathcal{K} $  given by \eqref{eq:defK}  and $ 
W= \tanh(\beta J) $ 

Taking the graph to be the two dimensional torus  
$ \G = \mathbb{T}^2_L $ of linear size $ L $,
it is most convenient to evaluate 
$\det (1 - \mathcal{KW} )$ not for the free boundary conditions operator, with which the relation \eqref{KW_det} holds in  its simplest form, but for its periodicized  version $\mathcal{KW}^{(\rm{per})}$.  
Under this change, the model ceases to be planar and the determinant is no longer directly giving the partition function.  We address this small discrepancy after an outline of the calculation of the determinant.  

The periodicised Kac-Ward operator commutes with lattice shifts.  Hence its  determinant factorizes into a product of the determinants of $4\times 4$ matrices, indexed by the Fourier modes, i.e. momenta $(k_1,k_2) \in \frac{2\pi}{L} (\Z _{\rm{mod } \, L})^2$.  
More explicitly, one may take as the basic building blocks  of $ \widehat{\mathcal E} $ the four edges $ \rightarrow = (0,(1,0) ) , \uparrow = (0,(0,1)), \leftarrow = (0,(-1,0)), \downarrow = (0,(0,-1)) $ which originate at the vertex $ 0 $.  All other edges  $ e \in \widehat{\mathcal E}$ result from translation of these four  and may hence be written as $e= (x , \alpha) $ with $ x \in \mathcal{V} $ and $  \alpha \in \{\rightarrow , \uparrow ,  \leftarrow, \downarrow \}$. 
 The space $ \ell^2(\widehat{\mathcal E}) $ decomposes into a direct sum of
 \be
 \mathcal{H}_k := \left\{ \psi(e) = e^{ik \cdot x} \, \left( \begin{matrix} \psi_\rightarrow \\ \psi_\uparrow \\ \psi_\leftarrow\\ \psi_\downarrow \end{matrix} \right) \,  \Bigg| \,  \psi_{\alpha} \in \C \right\} \, , \quad k \in \frac{2\pi}{L} \left(\mathbb{Z}_{\mod L}\right)^2 \, . 
 \ee
The matrix $ \mathcal{K}^{(\rm{per})}$ is block diagonal in this decomposition, and on $  \mathcal{H}_k $  it acts as
 \be
\mathcal{K}^{(\rm{per})}_k :=   \left( \begin{matrix} 1 & 0 & e^{-i\pi/4} & e^{i\pi/4} \\ 0 & 1 & e^{i\pi/4} & e^{-i\pi/4} \\ e^{i\pi/4} & e^{-i\pi/4} & 1 & 0 \\ e^{-i\pi/4} & e^{i\pi/4} & 0 & 1 \end{matrix} \right)   \left( \begin{matrix} e^{-ik_1} & 0 & 0 & 0 \\ 0 & e^{ik_1} & 0 & 0 \\ 0 & 0 & e^{ik_2} & 0 \\ 0 & 0 & 0 & e^{-ik_2} \end{matrix} \right)  \, . 
 \ee
 
Calculations of the corresponding $4\times 4$ determinants, and a bit of algebraic manipulation, then yield,  
\begin{multline}  \label{Riemann} 
\frac{1}{|\mathcal{V}|} \log \det\left( 1 - \mathcal{K}^{(\rm{per})}\mathcal {W}\right) \ = \  \frac{1}{2} \ln[ 2W (1-W^2)]  \ + \    \\ 
+ \ 
\frac{1}{L^2} \sum_{k \in \frac{2\pi}{L} \left(\mathbb{Z}_{\!\!\mod L} \right)^2 }  \ln \left\{ \left[  Y(\beta) + \frac{1}{Y(\beta)}  -2 \right ] + E(k) \right\} \, .  
\end{multline} 
with $Y(\beta) := \sinh(2\beta J ) $ and $ E(k) :=  2 \left [ \sin^2(k_1/2) + \sin^2(k_2/2)\right]$.   \\ 

At any    $\beta\neq \beta_c$, the latter defined by $Y(\beta_c)  = 1$,  
  the argument of the logarithm in \eqref{Riemann}  is  positive and continuous uniformly in $(k_1,k_2)$.  The sum can be viewed as a Riemann approximation of an  integral, and  one gets:  
\begin{multline} \label{eq:per}
\lim_{L \to \infty}  \frac{1}{L^2} \ln \det\left( 1 - \mathcal{K}^{(\rm{per})}\mathcal {W}_{[0,L]^2}\right) \ = \  
   \frac{1}{2}   \ln  \cosh( 2\beta J)   \ +    \\ 
+ \ \frac{1}{2} \int \ln \left\{ \left[Y(\beta)+ Y(\beta)^{-1} - 2 \right]+ E(k_1,k_2) \right\} \frac{dk_1 dk_2}{(2 \pi)^2}   
\, 
\end{multline}

Recalling that boundary conditions have only a surface effect on the energy density, it seems natural to expect that for  $\beta \neq \beta_c$ this would also be the case for the effect of the change in the above calculation from free to periodic boundary conditions on 
$\mathcal K$.   Under this assumption the above calculation yields Onsager's formula \eqref{eq:Onsager}.

The incompleteness of this argument is highlighted by
the observation that for $\beta = \beta_c$  the above computed determinant vanishes  due to the multiplicative contribution of the $k=(0,0)$ mode.      
The gap can be patched in a number of ways, of which we like to mention two (see also~\cite{PW55}).

First, a direct analysis of the free boundary condition determinant is also possible and is only a bit  more involved than the periodic case, cf.~\cite{KLM13}.  The two  are related by a spectral shift of order $O(L)$ which  for $\beta \neq \beta_c$ does not affect the result.   An alternative proof is outlined next. 

\section{Interpreting the periodic determinant} \label{sec:interp_per_bc}

 As shown drastically by the  vanishing of 
 $\det\left( 1 - \mathcal{K}^{(\rm{per})}\mathcal {W}\right) $, this  determinant is not a partition function (which cannot vanish at real couplings).   What it is can be answered with the help of Theorem~\ref{thm:beyondP}, from which we learn:   

\begin{figure}  \begin{center}
\includegraphics[width=0.35 \textwidth] {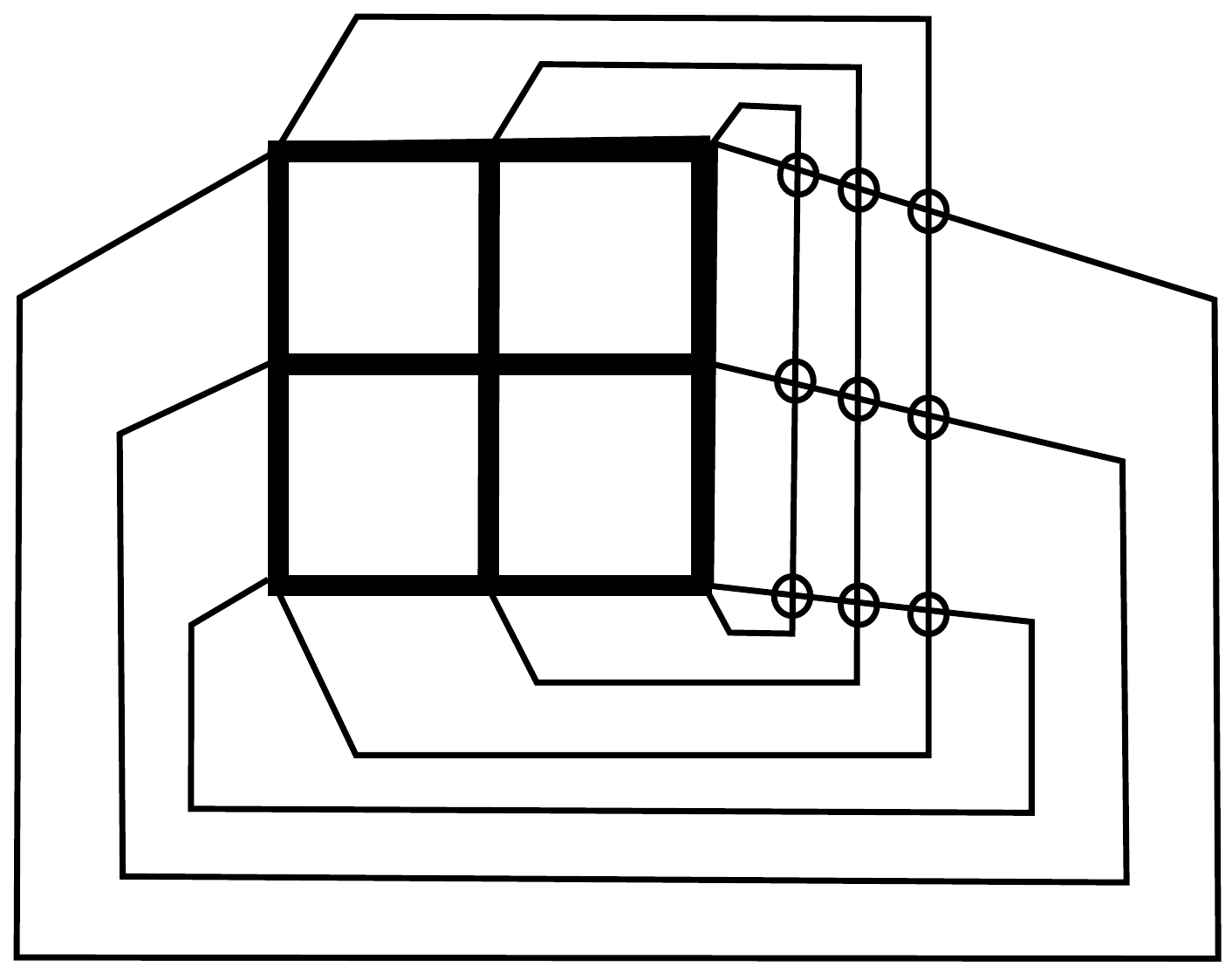} 
\caption{A graph supporting the $3\times3$ Ising spin system with periodic boundary condition.  Starting from the free boundary conditions periodicity is obtained by adding two sets of handles.  The circles mark their non-vertex crossings. }
\label{fig:per_bc}  
\end{center}
\end{figure}

\begin{lemma}  For the Ising model on $(0,L]^2 \cap \mathbb{Z}^2 $ with periodic boundary conditions, 
\begin{eqnarray}  \label{det_AP}  
\det\left( 1 - \mathcal{K}^{(\rm{per})}\mathcal {W}\right) & = &   
\sum_{\Gamma : \partial \Gamma =   \emptyset}   (-1)^{F_1(\Gamma) +  F_2(\Gamma)}  \, (-1) ^{ F_1(\Gamma) \cdot  F_2(\Gamma)} \, w(\Gamma) 
  \notag \\[2ex] 
& = &   \widehat {Z}_L^{\rm (per)}  \; 
\langle   (-1) ^{ F_1(\Gamma) +  F_2(\Gamma) + F_1(\Gamma) \cdot  F_2(\Gamma)} \rangle^{(per)} 
\end{eqnarray}  
where 
 $F_j(\Gamma)$ are the numbers of the 
the  horizontal ($j=1$) and  vertical ($j=2$) handles included in $\Gamma$  
(as indicated in Fig.~\ref{fig:per_bc}).  In the last expression,  the average is with respect to the probability distribution over the even subgraphs induced by the Gibbs state with periodic boundary conditions. 
\end{lemma} 

\begin{proof} 
Equation~\eqref{det_AP}  follows from \eqref{eq:beyondP}, for whose application it should be noted that 
\be 
 n_o(\Gamma) \ = \  F_1(\Gamma) \times   F_2(\Gamma)
\ee 
 and the Kac-Ward phase for each of the added handles  is 
 $e^{i (2\pi)/2} = -1$.  The product of these factors yield the oscillatory term in the first equation.   The second is just a recast of the first relation in terms of the corresponding  Gibbs state average. 
\end{proof} 

To shed more light on \eqref{det_AP} it helps to supplement it with rigorous results which are expressed in terms of bounds. 
For the binary term in \eqref{det_AP} to take the value  $(-1)$ it is required that  at least one of the $F_j(\Gamma)$ be odd.  However, under this condition $\Gamma$ has odd flux through each of the vertical (if $j=1$) or each of the horizontal (if $j=2$) vertex-avoiding  lines.  For even subgraphs $\Gamma$ each of these conditions implies that $\Gamma$ includes a connected path of edges of total length $L$.    
However, through the non-perturbative sharpness of phase transition theorem~\cite{ABF87,DRT17}  for all $\beta < \beta_{LRO}$, with the latter defined by the onset of   long range order (LRO),  the above event's probability  decays  exponentially in $L$.   
 Hence as $L\to \infty$ the last factor in $\eqref{det_AP}$ converges  to $1$, and consequently the limit computed in \eqref{eq:per} yields the  Onsager formula \eqref{eq:Onsager}.   By duality arguments the conclusion is valid  also for the dual range.  Taking into account the continuity of the free energy in $\beta$, we learn that   \eqref{eq:Onsager} is valid 
throughout the closed set $\R \backslash  (\beta_{\rm LRO} , \beta_{\rm LRO}^*)$.

  By other techniques -- outgrowths of either the breakthrough work of  Harris (cf.~\cite{G06}), or the Russo-Seymour-Welsh theory~\cite{DRT16} -- it is known that 
  the Ising model does not have a phase in which LRO is exhibited simultaneously by the model and its dual.  This implies that 
  \be
  \beta_{\rm LRO} = \beta_{\rm LRO}^* = \beta_c \, .
  \ee   
  Hence the above arguments prove \eqref{eq:Onsager}  for all $\beta \in \R$.

\section*{Acknowledgements}
This work was supported in part by the NSF grant DMS-1613296, the Weston Visiting Professorship at the Weizmann Institute (MA) and  a PU Global Scholarship (SW).  MA  thanks the Faculty of Mathematics and Computer Sciences and the Faculty of Physics  at WIS  for the hospitality enjoyed there. We thank Hugo Duminil-Copin and Vincent Tassion for the pleasure of collaboration on topics related to this work.

\end{document}